\documentclass[11pt,letterpaper]{article}

\def\showauthornotes{0}

\def\showkeys{0}
\def\showdraftbox{0}
\def\showcolorlinks{1}
\def\usemicrotype{0}
\def\showfixme{0}

\usepackage{etex}

\usepackage[l2tabu, orthodox]{nag}

\usepackage{xspace,enumerate}

\usepackage[dvipsnames]{xcolor}

\usepackage[american]{babel}

\usepackage{mathtools}

\usepackage{amsthm}

\usepackage[ruled,vlined,linesnumbered]{algorithm2e}
\SetKwRepeat{Do}{do}{while}

\newtheorem{theorem}{Theorem}[section]
\newtheorem*{theorem*}{Theorem}

\newtheorem*{proposition*}{Proposition}
\newtheorem{lemma}[theorem]{Lemma}
\newtheorem*{lemma*}{Lemma}
\newtheorem{corollary}[theorem]{Corollary}

\newtheorem*{conjecture*}{Conjecture}
\newtheorem{fact}[theorem]{Fact}
\newtheorem*{fact*}{Fact}

\newtheorem*{hypothesis*}{Hypothesis}

\theoremstyle{definition}
\newtheorem{definition}[theorem]{Definition}
\newtheorem*{definition*}{Definition}

\newtheorem*{problem*}{Problem}

\theoremstyle{remark}
\newtheorem{claim}[theorem]{Claim}
\newtheorem*{claim*}{Claim}
\newtheorem{remark}[theorem]{Remark}
\newtheorem*{remark*}{Remark}

\newtheorem*{observation*}{Observation}

\usepackage[
letterpaper,
top=1in,
bottom=1in,
left=1in,
right=1in]{geometry}

\usepackage{amsmath,amsfonts,amssymb}

\usepackage{newpxtext} 
\usepackage{textcomp} 
\usepackage[varg,bigdelims]{newpxmath}
\usepackage[scr=rsfso]{mathalfa}
\usepackage{bm} 
\let\mathbb\varmathbb

\ifnum\showkeys=1
\usepackage[color]{showkeys}
\fi

\ifnum\showcolorlinks=1
\usepackage[
colorlinks=true,
urlcolor=blue,
linkcolor=blue,
citecolor=OliveGreen,
]{hyperref}
\fi

\ifnum\showcolorlinks=0
\usepackage[
pagebackref,
colorlinks=false,
pdfborder={0 0 0}
]{hyperref}
\fi

\usepackage{prettyref}
\newcommand{\pref}{\prettyref}

\newcommand{\savehyperref}[2]{\texorpdfstring{\hyperref[#1]{#2}}{#2}}

\newrefformat{eq}{\savehyperref{#1}{\textup{(\ref*{#1})}}}
\newrefformat{prog}{\savehyperref{#1}{\textup{Program \ref*{#1}}}}
\newrefformat{eqn}{\savehyperref{#1}{\textup{(\ref*{#1})}}}
\newrefformat{lem}{\savehyperref{#1}{Lemma~\ref*{#1}}}
\newrefformat{def}{\savehyperref{#1}{Definition~\ref*{#1}}}
\newrefformat{thm}{\savehyperref{#1}{Theorem~\ref*{#1}}}
\newrefformat{cor}{\savehyperref{#1}{Corollary~\ref*{#1}}}
\newrefformat{cha}{\savehyperref{#1}{Chapter~\ref*{#1}}}
\newrefformat{sec}{\savehyperref{#1}{Section~\ref*{#1}}}
\newrefformat{app}{\savehyperref{#1}{Appendix~\ref*{#1}}}
\newrefformat{tab}{\savehyperref{#1}{Table~\ref*{#1}}}
\newrefformat{fig}{\savehyperref{#1}{Figure~\ref*{#1}}}
\newrefformat{hyp}{\savehyperref{#1}{Hypothesis~\ref*{#1}}}
\newrefformat{algo}{\savehyperref{#1}{Algorithm~\ref*{#1}}}
\newrefformat{rem}{\savehyperref{#1}{Remark~\ref*{#1}}}
\newrefformat{item}{\savehyperref{#1}{Item~\ref*{#1}}}
\newrefformat{step}{\savehyperref{#1}{step~\ref*{#1}}}
\newrefformat{conj}{\savehyperref{#1}{Conjecture~\ref*{#1}}}
\newrefformat{fact}{\savehyperref{#1}{Fact~\ref*{#1}}}
\newrefformat{prop}{\savehyperref{#1}{Proposition~\ref*{#1}}}
\newrefformat{prob}{\savehyperref{#1}{Problem~\ref*{#1}}}
\newrefformat{claim}{\savehyperref{#1}{Claim~\ref*{#1}}}
\newrefformat{relax}{\savehyperref{#1}{Relaxation~\ref*{#1}}}
\newrefformat{red}{\savehyperref{#1}{Reduction~\ref*{#1}}}
\newrefformat{part}{\savehyperref{#1}{Part~\ref*{#1}}}
\newrefformat{obs}{\savehyperref{#1}{Observation~\ref*{#1}}}
\newrefformat{corr}{\savehyperref{#1}{Corollary~\ref*{#1}}}

\newcommand{\Sref}[1]{\hyperref[#1]{\S\ref*{#1}}}

\usepackage{nicefrac}

\ifnum\usemicrotype=1
\usepackage{microtype}
\fi

\ifnum\showauthornotes=1
\newcommand{\Authornote}[2]{{\sffamily\small\color{red}{[#1: #2]}}}
\newcommand{\Authornotecolored}[3]{{\sffamily\small\color{#1}{[#2: #3]}}}
\newcommand{\Authorcomment}[2]{{\sffamily\small\color{gray}{[#1: #2]}}}
\newcommand{\Authorstartcomment}[1]{\sffamily\small\color{gray}[#1: }

\newcommand{\Authorfnote}[2]{\footnote{\color{red}{#1: #2}}}
\newcommand{\Authorfixme}[1]{\Authornote{#1}{\textbf{??}}}
\newcommand{\Authormarginmark}[1]{\marginpar{\textcolor{red}{\fbox{\Large #1:!}}}}
\else
\newcommand{\Authornote}[2]{}
\newcommand{\Authornotecolored}[3]{}
\newcommand{\Authorcomment}[2]{}
\newcommand{\Authorstartcomment}[1]{}

\newcommand{\Authorfnote}[2]{}
\newcommand{\Authorfixme}[1]{}
\newcommand{\Authormarginmark}[1]{}
\fi

\newcommand{\Pnote}{\Authornotecolored{ForestGreen}{P}}

\ifnum\showfixme=0

\fi

\usepackage{boxedminipage}

\newcommand{\set}[1]{\{#1\}}

\newcommand{\norm}[1]{\lVert#1\rVert}
\newcommand{\Norm}[1]{\left\lVert#1\right\rVert}

\newcommand{\iprod}[1]{\langle#1\rangle}

\newcommand{\Esymb}{\mathbb{E}}
\newcommand{\Psymb}{\mathbb{P}}
\newcommand{\Vsymb}{\mathbb{V}}

\DeclareMathOperator*{\E}{\Esymb}
\DeclareMathOperator*{\Var}{\Vsymb}
\DeclareMathOperator*{\ProbOp}{\Psymb}
\DeclareMathOperator*{\pE}{{\tilde\Esymb}}

\renewcommand{\Pr}{\ProbOp}

\DeclareMathOperator*{\argmin}{argmin}

\newcommand{\textparen}[1]{\text{(#1)}}

\ifx\because\undefined
\newcommand{\because}[1]{\textparen{because #1}}
\else
\renewcommand{\because}[1]{\textparen{because #1}}
\fi

\newcommand{\defeq}{\stackrel{\mathrm{def}}=}

\newcommand{\seteq}{\mathrel{\mathop:}=}

\newcommand{\mper}{\,.}
\newcommand{\mcom}{\,,}

\newcommand\bdot\bullet

\ifx\mathds\undefined 
\DeclareMathOperator{\Ind}{\mathbb{I}}
\else
\DeclareMathOperator{\Ind}{\mathds{1}}
\fi

\DeclareMathOperator{\Id}{\mathrm{Id}}

\DeclareMathOperator{\poly}{poly}

\newcommand{\etal}{et al.\xspace}

\newcommand{\Z}{\mathbb Z}
\newcommand{\N}{\mathbb N}
\newcommand{\R}{\mathbb R}

\newcommand{\cA}{\mathcal A}
\newcommand{\cB}{\mathcal B}
\newcommand{\cC}{\mathcal C}
\newcommand{\cD}{\mathcal D}
\newcommand{\cE}{\mathcal E}
\newcommand{\cF}{\mathcal F}

\newcommand{\cL}{\mathcal L}

\newcommand{\cN}{\mathcal N}

\newcommand{\cP}{\mathcal P}
\newcommand{\cQ}{\mathcal Q}
\newcommand{\cR}{\mathcal R}
\newcommand{\cS}{\mathcal S}
\newcommand{\cT}{\mathcal T}

\newcommand{\cX}{\mathcal X}
\newcommand{\cY}{\mathcal Y}

\usepackage{mathrsfs}

\newcommand{\bbP}{\mathbb P}

\ifnum\showdraftbox=1

\else

\fi

\let\epsilon=\varepsilon

\numberwithin{equation}{section}

\newcommand\MYcurrentlabel{xxx}

\newcommand{\MYstore}[2]{%
  \global\expandafter \def \csname MYMEMORY #1 \endcsname{#2}%
}

\newcommand{\MYload}[1]{%
  \csname MYMEMORY #1 \endcsname%
}

\newcommand{\MYnewlabel}[1]{%
  \renewcommand\MYcurrentlabel{#1}%
  \MYoldlabel{#1}%
}

\newcommand{\MYdummylabel}[1]{}

\newcommand{\torestate}[1]{%
  \let\MYoldlabel\label%
  \let\label\MYnewlabel%
  #1%
  \MYstore{\MYcurrentlabel}{#1}%
  \let\label\MYoldlabel%
}

\newcommand{\restatetheorem}[1]{%
  \let\MYoldlabel\label
  \let\label\MYdummylabel
  \begin{theorem*}[Restatement of \prettyref{#1}]
    \MYload{#1}
  \end{theorem*}
  \let\label\MYoldlabel
}

\newcommand{\restatedef}[1]{%
  \let\MYoldlabel\label
  \let\label\MYdummylabel
  \begin{definition*}[Restatement of \prettyref{#1}]
    \MYload{#1}
  \end{definition*}
  \let\label\MYoldlabel
}

\newcommand{\restatelemma}[1]{%
  \let\MYoldlabel\label
  \let\label\MYdummylabel
  \begin{lemma*}[Restatement of \prettyref{#1}]
    \MYload{#1}
  \end{lemma*}
  \let\label\MYoldlabel
}

\newcommand{\restateprop}[1]{%
  \let\MYoldlabel\label
  \let\label\MYdummylabel
  \begin{proposition*}[Restatement of \prettyref{#1}]
    \MYload{#1}
  \end{proposition*}
  \let\label\MYoldlabel
}

\newcommand{\restatefact}[1]{%
  \let\MYoldlabel\label
  \let\label\MYdummylabel
  \begin{fact*}[Restatement of \prettyref{#1}]
    \MYload{#1}
  \end{fact*}
  \let\label\MYoldlabel
}

\newcommand{\restateobs}[1]{%
  \let\MYoldlabel\label
  \let\label\MYdummylabel
  \begin{observation*}[Restatement of \prettyref{#1}]
    \MYload{#1}
  \end{observation*}
  \let\label\MYoldlabel
}
\newcommand{\restate}[1]{%
  \let\MYoldlabel\label
  \let\label\MYdummylabel
  \MYload{#1}
  \let\label\MYoldlabel
}

\let\origparagraph\paragraph
\renewcommand{\paragraph}[1]{\origparagraph{#1.}}

\allowdisplaybreaks

\usepackage[newenum,newitem,increaseonly]{paralist}
\usepackage{enumitem}

\usepackage{comment}

\let\citet\cite

\theoremstyle{definition}

\DeclareUrlCommand\email{}

\newcommand{\restateproblem}[2]{%
  \let\MYoldlabel\label
  \let\label\MYdummylabel
  \begin{problem*}[Restatement of \prettyref{#1}, {#2}]
    \MYload{#1}
    \end{problem*}
  \let\label\MYoldlabel
}
\usepackage{bm}

\usepackage{tikz,bbm}

\newcommand\whichfont{1}
\ifnum\whichfont=1
\newcommand{\Cov}{\boldsymbol{\mathrm{Cov}}}
\newcommand{\sumn}{\frac{1}{M}\sum\limits_{i=1}^N}
\newcommand{\pVar}{\widetilde{\mathbb{Var}}}
\newcommand{\pCov}{\widetilde{\mathbf{Cov}}}

\else
\newcommand{\Cov}{\mathbf{Cov}}
\newcommand{\pCov}{\widetilde{\mathbf{Cov}}}
\fi

\usepackage{turnstile}
\newcommand{\SoSp}[1]{\sststile{#1}{}}

\usepackage{thmtools}
\usepackage{thm-restate}

\usepackage{hyperref}

\usepackage{cleveref}

\newcommand{\hu}{\hat{\mu}}
\newcommand{\pe}{\tilde{\mathbb{E}}}

\newcommand{\calA}{\mathcal A}

\newcommand{\calG}{\mathcal G}
\newcommand{\calH}{\mathcal H}

\newcommand{\calQ}{\mathcal Q}
\newcommand{\calR}{\mathcal R}

\newcommand{\calX}{\mathcal X}
\newcommand{\calY}{\mathcal Y}
\newcommand{\calZ}{\mathcal Z}

\newcommand{\bcalA}{\boldsymbol{\calA}} 
\newcommand{\bcalG}{\boldsymbol{\calG}} 
\newcommand{\bcalH}{\boldsymbol{\calH}} 
\newcommand{\bcalQ}{\boldsymbol{\calQ}} 
\newcommand{\bcalR}{\boldsymbol{\calR}} 
\newcommand{\bcalX}{\boldsymbol{\calX}} 
\newcommand{\bcalY}{\boldsymbol{\calY}} 
\newcommand{\bcalZ}{\boldsymbol{\calZ}}

\usepackage{tikz}

\usepackage{relsize}
\usepackage[font=footnotesize]{caption}
\usepackage[flushleft,para]{threeparttable}
\makeatletter
\g@addto@macro\TPT@defaults{\footnotesize}
\makeatother

\DeclareUrlCommand\email{}

\renewcommand{\it}{\em}

\DeclareUrlCommand\email{}

\newcommand{\SoSle}{\preceq}
\newcommand{\SoSge}{\succeq}

\let\pref=\prettyref

\begin{document}

\title{List Decodable Learning via Sum of Squares}
\author{Prasad Raghavendra\thanks{University of California, Berkeley, research supported by NSF Grant CCF 1718695.} \and Morris Yau \thanks{University of California, Berkeley, research supported by NSF Grant CCF 1718695.}}
\maketitle
\thispagestyle{empty}

\begin{abstract}
In the list-decodable learning setup, an overwhelming majority (say a $1-\beta$-fraction) of the input data consists of outliers and the goal of an algorithm is to output a small list $\cL$ of hypotheses such that one of them agrees with inliers.
We develop a framework for list-decodable learning via the Sum-of-Squares SDP hierarchy
 and demonstrate it on two basic statistical estimation problems  
\begin{itemize}
 \item  {\it Linear regression:}  Suppose we are given labelled examples $\{(X_i,y_i)\}_{i \in [N]}$ containing a subset $S$ of $\beta N$ {\it inliers} $\{X_i \}_{i \in S}$ that are drawn i.i.d. from standard Gaussian distribution $N(0,I)$ in $\R^d$, where the corresponding labels $y_i$ are well-approximated by a linear function $\ell$. We devise an algorithm that outputs a list $\cL$ of linear functions such that there exists some $\hat{\ell} \in \cL$ that is close to $\ell$.
  
  This yields the first algorithm for linear regression in a list-decodable setting.  Our results hold for any distribution of examples whose concentration and anticoncentration can be certified by Sum-of-Squares proofs.

 \item {\it Mean Estimation:} 
  Given data points $\{X_i\}_{i \in [N]}$ containing a subset $S$ of $\beta N$ {\it inliers} $\{X_i \}_{i \in S}$ that are drawn i.i.d. from a Gaussian distribution $N(\mu,I)$ in $\R^d$, we devise an algorithm that generates a list $\cL$ of means such that there exists $\hat{\mu} \in \cL$ close to $\mu$.

The recovery guarantees of the algorithm are analogous to the existing algorithms for the problem by 
  Diakonikolas \etal \cite{diakonikolas2018list} and Kothari \etal \cite{kothari2017better}.
 \end{itemize}
 
 In an independent and concurrent work, Karmalkar \etal \cite{KlivansKS19} also obtain an algorithm for list-decodable linear regression using the Sum-of-Squares SDP hierarchy.
 
\end{abstract}

\clearpage

\tableofcontents
\clearpage

\section{Introduction}
The presence of outliers in data poses a fundamental challenge to algorithms for high-dimensional statistical estimation.
%
While robust statistics have been explored extensively for several decades now \cite{huber2011robust}, a flurry of recent work starting with \cite{klivans2009learning, awasthi2014power, lai2016agnostic,diakonikolas2016robust} have led to new robust algorithms for high-dimensional  statistical tasks such as mean estimation, covariance estimation, linear regression and learning linear separators.

More recently, a promising line of work \cite{hopkins2018mixture,kothari2017outlier,kothari2017better,klivans2018efficient} has brought to bear the sum-of-squares SDP hierarchy on problems from robust statistics, resulting in new algorithms under fairly minimal assumptions.  
Continuing this line of inquiry, we further develop the SoS SDP based approach to robust statistics.  Specifically, we develop a framework for list-decodable learning via the SoS SDP hierarchy.  
We demonstrate the framework by devising the first polynomial-time algorithm for linear regression that can extract an underlying linear function even in the presence of an overwhelming majority of outliers.

Linear regression is a corner-stone problem statistics and the underlying optimization problem is perhaps the central example of convex optimization.
In the classical setup for linear regression, the input data consists of labelled examples $\{(X_i,y_i)\}_{i \in [N]}$ where  $\{X_i\}_{i \in [N]}$ are drawn i.i.d. from a distribution $\cD$ over $\R^d$, and the labels $\{y_i\}_{i \in [N]}$ are noisy evaluations of a linear function.
Specifically, the labels $y_i$ are given by $y_i = \iprod{\hat{\ell},X_i} + \gamma_i$ where $\gamma_i$ denotes the noise.
The goal is to recover an estimate $\ell$ to the linear function $\hat{ell}$.

In its simplest form, the distribution $\cD = N(0,\Id)$ is the standard Gaussian measure, the noise $\gamma_i$ is mean zero and independent of the example $X_i$.
The linear function $\hat{ell}$ can be recovered (up to statistical deviations) by minimizing the squared loss namely,
$$ \ell = \argmin \E_{(X,y) \sim \cD}[ (\iprod{\ell,X}-y)^2] \mper $$
From an algorithmic standpoint, the realizable setting of linear regression is fairly well understood.

The focus of this work is on algorithms for linear regression that are robust to the presence of outliers.
While there is an extensive literature on robust linear regression (see  \cite{rousseeuw2005robust,bhatia2015robust,bhatia2017consistent} and the references therein), there are no algorithms that are robust to an overwhelming majority of outliers.
Concretely, consider the following problem setup:  we are given labelled examples $\{(X_i,y_i)\}_{i \in [N]}$ such that a $\beta$-fraction of these examples are drawn from the underlying distribution, while the remaining $(1-\beta)$-fraction of examples are adversarially chosen.  Formally, let us suppose $\beta N$ examples are drawn from the distribution with $X \sim N(0,\Id)$ and $y = \hat{\ell}(X)+\gamma$, while the rest of the examples are arbitrary. 

For $\beta < \frac{1}{2}$, it is information theoretically impossible to estimate the linear function $\hat{\ell}$, since the input data can potentially be consistent with $\frac{1}{\beta}$-different linear functions $\ell$.
It is natural to ask if an efficient algorithm can recover a small list of candidate linear functions $\cL = \{ \ell_1,\ldots,\ell_t\}$ such that one of them is close to $\ell$.
The learning model such as the one above where the goal of the algorithm is to find a small list of candidate hypotheses is referred to as {\it list-decodable} learning.

This model was introduced by Balcan \etal \cite{BalcanBV08} in the context of clustering, and has been the subject of a line of work \cite{CharikarSV17,diakonikolas2018list,SteinhardtVC16, SteinhardtKL17,kothari2017outlier} in the recent past.
The problem of linear regression in the setup of list-decodable learning had remained open.  

\paragraph{Our Results}

In this work, we use the sum-of-squares SDP hierarchy to devise an efficient algorithm for the list-decodable linear regression problem.
Formally, we show the following result.
\begin{theorem}\label{thm:main}
    There is an algorithm $\cA$  such that the following holds for every $\beta > 0$.
    
    Suppose for a sequence of labelled examples $\{ (X_i,y_i) \}_{i \in [N]}$, there exists a linear function $\hat{\ell}(X) = \iprod{\hat{\ell},X}$ and a subset $S \subset [N]$ of $\beta N$ examples such that,
    \begin{enumerate}
        \item For any $\epsilon > 0$, the $k^{th}$ emprical moments of $\{X_i\}_{i \in S}$ are close to that of the underlying distribution of examples for each $k = 1 \ldots K$ i.e., 
            $$ \Norm{ \E_{i \in S} X_i^{\otimes k} - \E_{X \sim \cD} X^{\otimes k} }_2 \leq \epsilon $$
            for some $K = O(1/\beta^4)$.
        \item  The injective tensor norm of the covariates is bounded.  That is to say for all degree $D \geq 4$ pseudoexpectations $\pE$ over indeterminate $v_1,...,v_d$ the fourth injective tensor norm 
        $$\sststile{4}{} \E_{X \sim \cD} \iprod{X,v}^4 \leq B\norm{v}^4$$
        is certifiably upper bounded by a constant $B$.  For standard gaussians $B = 3$.  More generally, any distribution satisfying a \textit{poincare} inequality has certifiably upper bounded injective tensor norms, see \cite{kothari2017better}.   
        \item The empirical loss of $\hat{\ell}$ is $(2,4)$-hypercontractive, i.e.,
            $$ \E_{i \in S} (y_i - \iprod{\hat{\ell}, X_i})^4  \leq g \cdot \left( \E_{i \in S} (y_i - \iprod{\hat{\ell},X_i})^2\right)^2 $$
            for some constant $g$
    \end{enumerate}
  Then the algorithm $\cA$ running on the set of examples $\{(X_i,y_i)\}_{i \in [N]}$ outputs a list of $O\left(\left(\frac{\norm{\hat{\ell}}}{\sigma}\right)^{\log 1/\beta}\right)$ candidate linear functions $\cL$ such that there exists $\ell \in \cL$ satisfying
        $$ \norm{\ell -\hat{\ell}}_2 \leq O\left(\frac{\sigma}{\beta^{3/2}}\right) $$
where $\sigma^2 \defeq \E_{i \in S} (y_i - \iprod{\hat{\ell},X_i})^2 $.  The runtime of the algorithm $\cA$ is $\left(\frac{\norm{\hat{\ell}}}{\sigma}\right)^{\log 1/\beta} \cdot N^{O(1/\beta^4)}$ for $N = d^{O(\frac{1}{\beta^4})}$.
\end{theorem}
Even in the absence of outliers, the information theoretic limit on the accuracy $\norm{\ell - \hat{\ell}} = \Omega(\sigma)$.  To interpret the list size and runtime bounds, consider the setting $\beta = 1/4$, $\norm{\hat{\ell}} = 1$ and noise rate $\sigma = 10^{-7}$.  In this case, the linear function $\hat{\ell}$ is specified by an arbitrary point in the unit ball, and the algorithm finds a constant-sized list such that one of the points $\ell$ in the ball satisfies $\norm{\ell-\hat{\ell}} < 0.01$.
More generally, $\left(\frac{\norm{\hat{\ell}}}{{\sigma}}\right)^d$ would be the size of a $\sigma$-net for the ball of radius $\norm{\hat{\ell}}$, but the list size is a fixed polynomial in $ \frac{\norm{\hat{\ell}}}{{\sigma}}$.

Our results on linear regression apply to a broader class of probability distributions on examples we term {\it "SoS certifiably anti-concentrated"} (see \pref{def:anticoncentration}).
Informally, these are probability distributions $\cD$ that admit an sum-of-squares proof of their anti-concentration along every direction.

Sum-of-Squares SDPs yield a unified framework for statistical estimation tasks \cite{hopkins2018mixture,kothari2017outlier,kothari2017better,klivans2018efficient} through the notion of identifiability proofs.
Roughly speaking, if there exists a sum-of-squares proof that the statistical parameter of interest is {\it identifiable} from the data, then the sum-of-squares SDP can be utilized to estimate the statistic.
In the setting of list decodable learning, the parameter of interest 
say, the underlying linear function is not uniquely determined by the data, thus breaking the paradigm of SoS proof of identifiability.
Alternately, the SoS SDP solution is potentially a convex combination of the different hypotheses on the list.  Therefore, a list-decodable learning algorithm via SoS SDP will have to involve some randomized rounding to isolate one hypotheses from the mixture.

We use the technique of conditioning \cite{DBLP:conf/focs/BarakRS11, DBLP:conf/soda/RaghavendraT12} to randomly isolate one hypothesis from the SoS SDP solution. 
More precisely, our algorithm iteratively conditions the SoS SDP solution on specific data points being inliers.  
The analysis of the algorithm argues that after conditioning on a small number of appropriately chosen data points being inliers, the SoS SDP solution is more or less supported on a unique hypothesis, that we can output.

The framework of rounding by iterative conditioning can also be applied to list-decodable mean estimation problem.  In the mean-estimation problem, a $\beta$-fraction of inliers in a set of $N$ data points $\{ X_1,\ldots, X_N\}$ are sampled from a distribution $\cD$.  The goal is to recover a list of points $\{\hat{\mu}_1,\ldots, \hat{\mu}_t \}$ such that one of them is close to mean of the inliers.  Diakonikolas \etal \cite{diakonikolas2018list} devise an algorithm for the problem when $\cD$ is a spherical Gaussian, while Kothari and Steinhardt \cite{kothari2017outlier} solve it on a broader class of distributions referred to as SoS-certifiable distributions.

A probability distribution $\cD$ is $(2k,B)$-SoS certifiable if the polynomial inequality $\E_{X \sim \cD}[\iprod{v,X - \E[X]}^{2k}] \leq B^{2k} \norm{v}^{2k}$ admits a sum-of-squares proof.  Similarly, an empirical distribution $\{X_i\}_{i \in S}$ is said to be $(2k,B)$-SoS certifiable if $\E_{i \in S} [\iprod{v,X_i - \E[X_i]}^{2k}] \leq B^{2k} \norm{v}^{2k}$ admits a sum-of-squares proof.
By applying our framework, we recover an algorithm for list-decodable mean estimation for SoS-certifiable distributions analogous to the work of \cite{kothari2017outlier}.    
Formally, we show the following.

\begin{theorem} \label{thm:main-mean}
 There is an algorithm $\cA$  such that the following holds for every $\beta > 0$.
Suppose for a sequence of labelled examples $\{ X_i \}_{i \in [N]}$, there exists a subset $S \subset [N]$ of $\beta N$ examples such that
the empirical distribution $\{X_i \}_{i \in S}$ is a $(2k,B)$-SoS certifiable then the algorithm $\cA$ returns a list of $L$ points of length $\poly(\frac{1}{\beta})$ such that there exists a point $\mu \in \cL$ with $\norm{\mu - \E_{i \in S} X_i} \leq O(\frac{B}{\beta^{1/k}})$. 
The runtime of the algorithm is $d^{\poly(\frac{1}{\beta})}$

\end{theorem}

%
%

\paragraph{Related Work}

\paragraph{List Decodable Learning}
Balcan \etal \cite{BalcanBV08} introduced the notion of list-decodable learning, specifically, the notion of list-clustering.  Charikar \etal \cite{CharikarSV17} formally defined the notions of list-decodable learning and semi-verified learning, and showed that learning problems  in the two models reduce to one another.  
Charikar \etal \cite{CharikarSV17} obtained algorithms for list-decodable learning in the general setting of stochastic convex optimization, and applied the algorithm to a variety of settings including mean estimation, density estimation and planted partition problems (also see \cite{SteinhardtVC16, SteinhardtKL17}).
%
%
The same model of {\it list-decodable learning} has been studied for the case of mean estimation \cite{kothari2017outlier} and Gaussian mixture learning \cite{kothari2017better,diakonikolas2018list}.

\paragraph{Linear Regression}

Several heuristics have been developed for modifying the ordinary least squares objective with the intent of minimizing the effect of outliers (see \cite{rousseeuw2005robust}).

Often, the terminology of “robust regression” is used to refer to a more restricted noise model where only the labels are allowed
to be corrupted adversarially (see \cite{nguyen2013exact,nguyen2013robust, bhatia2015robust, bhatia2017consistent}).
The work of Bhatia \etal \cite{bhatia2017consistent} yields an algorithm for regression when the noise introduced is oblivious to the examples, but with a desirable property called consistency, in that the error rate approaches zero with increasing sample size.

There are several works on regression in the more stringent noise models.
Balakrishnan \etal \cite{balakrishnan2017computationally} devise algorithms for sparse linear regression in Huber's contamination model.
Diakonikolas \etal \cite{diakonikolas2019efficient} and Klivans \etal \cite{klivans2018efficient} yield algorithms in the most stringent noise models where both the examples and the labels can be arbitrarily corrupted.
The latter appeals to SoS SDP relaxations and is applicable to a broad class of distributions under very minimal assumptions.
All of the work described above apply at small noise rates, where the total fraction of corruptions are bounded by a small constant.

In a setting where the outliers are an overwhelming majority of the samples, linear regression algorithms have been studied for recovering a sparse vector $x$ \cite{wright2010dense, nguyen2013exact,nguyen2013robust}.

Finally, Hardt and Moitra \cite{hardt2013algorithms} consider a related problem of robust subspace recovery where a fraction of the samples lie within a $d$-dimensional subspace of $\R^n$.  They devise an algorithm when there are at least $d/n$-fraction of inliers (which corresponds to $(1-1/n)$-fraction of inliers for linear regression).  Furthermore, they show that if we make no additional distributional assumptions on the points, then it is computationally intractable to solve robust subspaec recovery with less than  $d/n$-fraction of inliers under the Small-Set Expansion Hypothesis.

In an independent and concurrent work, Karmalkar \etal \cite{KlivansKS19} also devise algorithms for list-decodable linear regression using the sum-of-squares SDP hierarchy.  The runtime and recovery guarantees of the algorithm are very similar to this work.  

\section{Preliminaries}
\label{sec:prelims}
\subsection{Proofs to Algorithms: Identifiability and Why it Fails}

At a high level, The proofs to algorithms method sets up a system of polynomial equalities and inequalities $\cP = \{f_1(x) = 0,f_2(x) = 0,...,f_m(x) = 0, g_1(x) \geq 0, g_2(x) \geq 0, ..., g_n(x) \geq 0\}$ and aims to output a solution $\theta$ to $\cP$.  Here we think of $\theta$ as a statistical parameter which in our case is either a mean estimate or a hyperplane.  In general, this is too much to ask for as the solution set of $\cP$ may be nonconvex and admit no discernible structure.  The SoS hierarchy is a powerful tool in convex optimization, designed to approximately solve polynomial systems.  The hierarchy is parameterized by its 'degree' $k$.  The degree corresponds to the size of a Semidefinite Program (SDP) used to solve for solutions to $\cP$.  The hope is that with higher degree, larger SDP's can obtain sharper approximations to $\theta$.  Thus, an immediate hurdle in designing efficient algorithms is to control 'k' with respect to the desired approximation guarantee.       

In general, outputting a solution $\theta$ even approximately is still too much to ask for.  Instead the SoS algorithm aims to output a fake distribution or 'pseudodistribution' over solutions to $\cP$.  Furthermore the SoS algorithm returns only the degree up to $k$ moments of the pseudodistribution $\zeta$.  That the pseudodistribution is not a true distribution lies at the heart of obtaining computationally efficient algorithms from SoS.  Thus, it can be said that pseudodistributions are relaxations of actual probability distributions over the solution set of $\cP$.  We will defer discussion of pseudodistributions and their dual objects pseudoexpectations to section \ref{SoS-toolkit}.     

In the context of unsupervised learning the goal is to estimate a parameter $\theta'$ from samples $x_1,...,x_n$. Identifiability refers to the property that any solution $\theta$ to $\cP$ is close to the true parameter $\theta'$, i.e $\norm{\theta - \theta'} << \text{small}$. Furthermore, if this proof of identifiability is captured by a sufficiently simple proof (a low degree SoS) then up to rounding issues $\theta$ can be found efficiently.  This paradigm has been immensely successful in designing SoS algorithms in machine learning settings.   

A key challenge for the list decoding problem is that even if it were possible to output a true distribution $\upsilon$ over solutions to $\cP$, another $\beta$ fraction of the dataset can imitate a solution to $\cP$.  Thus, direct identifiability fails.  A natural fix would be to sample from the distribution $\upsilon$ in the hopes of finding a region of substantial probability mass around $\theta'$.  The analogue of sampling on pseudodistributions is 'rounding'.  The core technical contribution of this work is developing new techniques for rounding pseudodistributions for high dimensional parameter estimation tasks.  Our method 'concentration rounding' has its roots in conditioning SoS SDP's, see \cite{DBLP:conf/focs/BarakRS11, DBLP:conf/soda/RaghavendraT12}. 
Next we present some standard tools when working with SoS and some properties of conditional pseudoexpectation.  

\subsection{Sum-of-Squares Toolkit} \label{SoS-toolkit}
\paragraph{Sum-of-Squares Proofs}

Fix a set of polynomial inequalities $\cA = \{ p_i(x) \geq 0 \}_{i \in [m]}$ in variables $x_1,\ldots,x_n$.
Starting with these ``axioms
$\cA$, a sum-of-squares proof of $q(x) \geq 0$ is given by an identity of the form,
\[
    \left(1+ \sum_{k \in [m']} b_k^2(x)\right) \cdot q(x) = \sum_{j\in [m'']} s_j^2(x) + \sum_{i \in [m]} a_i^2(x) \cdot p_i(x) \mcom
\]
where $\{s_j(x)\}_{j \in [m'']},\{a_i(x)\}_{i \in [m]}, \{b_k(x)\}_{i \in [m']}$ are real polynomials.
It is clear that any identity of the above form manifestly certifies that the polynomial $q(x) \geq 0$, whenever each $p_i(x) \geq 0$ for real $x$.
The degree of the sum-of-squares proof is the maximum degree of all the summands, i.e., $\max \{\deg(s_j^2), \deg(a_i^2 p_i)\}_{i,j}$.

Sum-of-squares proofs extend naturally to polynomial systems that involve a set of equalities $\{r_i(x) = 0\}$ along with a set of inequalities $\{ p_i(x) \geq 0\}$.
We can extend the definition syntactically by replacing each equality $r_i(x) = 0$ by a pair of inequalities $r_i(x) \geq 0$
and $-r_i(x) \geq 0$.

We will the use the notation $\cA \sststile{d}{x} \set{q(x) \geq 0}$ to denote that the assertion that,  there exists a degree-$d$ sum-of-squares proof of $q(x) \geq 0$ from the set of axioms $\cA$.
The superscript $x$ in the notation $\cA \sststile{d}{x} \set{q(x) \geq 0}$ indicates that the sum-of-squares proof is an identity of polynomials where $x$ is the formal variable.

A useful quality of SoS proofs is that they can be composed in the following sense.
\begin{fact}
For polynomial systems $\cA$ and $\cB$, if $\cA \sststile{d} \{p(x) \geq 0\}$ and $\cB \sststile{d'} \{q(x)\geq 0\}$ then $\cA \cup \cB \sststile{\max(d,d')}\{p(x) + q(x) \geq 0\}$.  Also $\cA \cup \cB \sststile{dd'} \{p(x)q(x) \geq 0\}$ 
\end{fact}
We now turn to pseudoexpectations, the dual object to SoS proofs.  
%
%

\paragraph{Pseudoexpectations}
%
\begin{definition}
Fix a polynomial system $\cP$ in $n$ variables $x \in \R^n$ consisting of inequalities $\{p_i(x) \ge 0\}_{i\in[m]}$.
A degree-$d$ pseudoexpectation $\pE : \R[x]_{\leq d} \to \R$ {\it satisfying $\cP$} is a linear functional over polynomials of degree at most $d$ with the properties that $\pE[1] = 1$, $\pE[p(x) a^2(x)] \geq 0$ for all $p \in \cP$ and polynomials $a$ such that $\deg(a^2 \cdot p) \le d$, and $\pE[q(x)^2] \ge 0$ whenever $\deg(q^2) \le d$.
\end{definition}
The properties above imply that when $\cA \sststile{d}{x} \set{q(x) \ge 0}$, then if $\pE$ is a degree-$d$ pseudoexpectation operator for the polynomial system defined by $\cA$, $\pE[q(x)] \ge 0$ as well.
This implies that $\pE$ satisfies several useful inequalities; for example, the Cauchy-Schwarz inequality.
(See e.g. \cite{DBLP:conf/stoc/BarakBHKSZ12} for details.)

\paragraph{SoS Algorithm}

The degree $D$ moment tensor of a pseudoexpectation $\pE_{\zeta}$ is the tensor $\pE_\zeta(1, x_1, x_2, . . . , x_n)^{\otimes D}$.
Each entry corresponds to the pseudo-expectation of all
monomials of degree at most $D$ in x. The set of all degree-$D$ moment tensors of degree $D$ pseudoexpectations is convex, and there's a separation oracle that runs in time $n^{O(D)}$.

\begin{fact} (\cite{Nesterov00}, \cite{Parrilo00}, \cite{Lasserre01}, \cite{Shor87}). For any $n$, $D \in \Z^+$, let $\pE_{\zeta}$ be degree $D$ pseudoexpectation satisfying a polynomial system $\cP$.  Then the following set has a $n^{O(D)}$-time weak
separation oracle (in the sense of \cite{GLS1981}):
 $$\{\pE_\zeta(1, x_1, x_2, . . . , x_n)^{\otimes D}| \text{ degree } D \text{ pseudoexpectations } \pE_{\zeta} \text{ satisfying }\cP \}$$

Armed with a separation oracle, the ellipsoid algorithm finds a degree $D$ pseudoexpectation in time $n^{O(D)}$, which we call the degree $D$ sum-of-squares algorithm. 
\end{fact}
Next we present some useful inequalities for working with SoS proofs and pseudoexpectations.  

\paragraph{Useful Inequalities}

\begin{fact} (Cauchy Schwarz)
Let $x_1,..,x_n,y_1,...,y_n$ be indeterminates, than 

$$\sststile{4}{} \big(\sum_{i \leq n}x_iy_i\big)^2 \leq \big(\sum_{i \leq n}x_i^2\big)\big(\sum_{i \leq n}y_i^2\big) $$
\end{fact}

\begin{fact} (Triangle Inequality) 
Let $x,y$ be $n$-length vectors of indeterminates, then 
$$\sststile{2}{} \norm{x + y}^2 \leq 2\norm{x}^2 + 2\norm{y}^2 $$
\end{fact}

\begin{fact}\label{fact:moment} (Moment Bounds)
Let $u = (u_1, . . . , u_k)$ be a vector of indeterminants. Let $\cD$ be Gaussian with
variance proxy 1. Let $t \geq 0$ be an integer. Then we have
$$\sststile{2t}{} \E_{x \sim \cD} \iprod{X,u}^{2t} \leq (2t)! \norm{u}^{2t}$$

\end{fact}
\Pnote{the above is not true for Sub Gaussian , only true for Gaussian}

\begin{fact}(Pseudoexpectation Cauchy Schwarz). 
Let $f(x)$ and $g(x)$ be degree at most $\ell \leq \frac{D}{2}$ polynomial in indeterminate $x$, then $$\pE[f(x)g(x)]^2 \leq \pE[f(x)^2]\pE[g(x)^2]$$  
\end{fact}

\begin{fact} (Pseudoexpectation Holder's)
Let p be a degree $\ell$ sum of squares polynomial, $t \in \mathbb{N}$, and $\tilde{\mathbb{E}}$ a degree $O(t \ell)$ pseudoexpectation. Then
$$\pE p(x)^{t-2} \leq \big( \pE p(x)^t \big)^{\frac{t-2}{t}} $$
\end{fact}

\begin{fact}(SoS Holder)
Let $X_1,..,X_n$ and $w_1,...,w_n$ be indeterminates. Let $q \in \N$ be a power of $2$, then
$$\set{w_i^2 = w_i \forall i \in [n]}\sststile{O(q)}{} \big( \sum_{i \leq n}w_iX_i\big)^q \big( \sum_{i \leq n} X_i^q\big)$$
and 
$$\set{w_i^2 = w_i \forall i \in [n]}\sststile{O(q)}{} \big( \sum_{i \leq n}w_iX_i\big)^q \big( \sum_{i \leq n} w_iX_i^q\big)$$
\end{fact}

\begin{fact} (Spectral Bounds)
Let $A \in \R^{d \times d}$ be a positive semidefinite matrix with $\lambda_{max}$ and $\lambda_{min}$ being the  largest and smallest eigenvalues of $A$ respectively. Let $\pE$ be a pseudoexpectation with degree greater than or equal to $2$ over indeterminates $v = (v_1,...,v_d)$.  Then we have 
$$\sststile{2}{} \iprod{A,vv^T} \leq \lambda_{max} \norm{v}^2 $$
and 
$$\sststile{2}{} \iprod{A,vv^T} \geq \lambda_{min} \norm{v}^2 $$
\end{fact}
\begin{remark}
We will make use of the following notation for $\norm{\cdot}_{op}$ and $\norm{\cdot}_{nuc}$ for operator and nuclear norm.
\end{remark}

Next we will discuss some useful properties of conditional pseudoexpectation.   
\paragraph{Conditional Pseudoexpectation}
\begin{definition} (Conditioning)
Given a degree $D$ pseudoexpectation operator \newline $\pE:\R[x_1,\ldots,x_n]^{\leq D}\to\R$ and a polynomial $e(x)$ of degree $d < D/2$, the conditioned pseudo-expectation operator $\pE\big|_{e}$ is given by,
$$ \pE[s(x)|e(x)] = \frac{\pE[s(x) e^2(x)]}{\pE[e^2(x)]} \mper$$
 $\pE\big|_{e}$ is a degree $D-2d$ pseudoexpectation functional that satisfies the same polynomial constraints as $\pE$.
\end{definition}

For a indeterminate $w$ satisfying the boolean constraint $w^2 = w$, we will use $\pE[\cdot |w]$ to denote the conditioned functional $\pE_{|w^2}$.
Given a degree $D$ pseudo-expectation operator $\pE$ satisfies a polynomial system $\cP$ all of whose polynomials are of degree at most $d$, for every polynomial $e$ of degree $\leq \frac{D-d}{2}$, the conditioned pseudoexpectation functional $\pE[ \cdot |e]$ also satisfies the system $\cP$,

For any two polynomials $p,q$ we define the pseudovariance as $$\pVar[p] \defeq \pE[p^2] - \pE[p]^2$$ and pseudo-covariance as  $$\pCov[p,q] \defeq \pE[p,q] - \pE[p]\pE[q]$$   We will also be making extensive use of the conditional pseudoexpectation toolkit.
\begin{fact} (Conditional Pseudoexpectation Toolkit)

Let $w$ satisfy the boolean constraint $w^2 = w$.  For a random variable $b$ taking values in $\{0,1\}$ such that $\Pr[b = 1] = \pE[w]$ and $\Pr[b = 0] = \pE[1 - w]$ we have the following useful facts

\begin{enumerate}
    \item (Law of Total Pseudoexpectation) $\E_b\pE[p(x,w)|w = b] = \pE[p(x,w)]$
    \item (Law of Total Pseudovariance) $\pVar[p(x,w)] - \E_{b}\pVar[p(x,w)|w = b] = \Var_b[\pE[p(x,w)|w = b]]$\\
    \item (Correlation Identity) $\pE[p(x,w)|w = b] = \frac{\pCov[p(x,w),w]}{\pVar[w]}b + \left(\pE[p(x,w)] - \frac{\pCov[p(x,w),w]}{\pVar[w]}\pE[w]\right)$
\end{enumerate}

\end{fact}

\begin{remark}
(Numerical accuracy). To make our error guarantees precise, we have to discuss issues of bit complexity.  
The SoS algorithm obtains a degree $D$ pseudoexpectation $\pE_\zeta$ satisfying polynomial system $\cP$ approximately.  That is to say, for every $g$ a sum of squares and $p_1, . . . , p_\ell \in \cP$
with deg$[g \prod p_i \leq D]$, one has $\pE g\prod_{i \in \cP} p_i \geq - 2^{-\Omega(n)} \norm{g}$, where $\norm{g}$ is $\ell_2$ norm of the coefficients
of $g$.  So long as all of the polynomials involved in our SoS proofs have coefficients bounded by $n^B$ for some large constant $B$, then for any polynomial system $\cQ$ such that $\cP \sststile{}{} \cQ$, we have $\cQ$ satisfied up to error $2^{-\Omega(n)}$.
\end{remark}

\subsection{Organization}

In \pref{sec:reg-overview} we go over the main ideas of list decoding robust regression for the covariate distribution $\cN(0,I)$.  Then in section \pref{sec:iteration} we wrap our algorithms in an iterative framework for sharp error guarantees.  In \pref{sec:conditioning} we prove the lemmas relevant to conditioning SoS SDP's.  In \pref{sec:reg-rounding} we present the exhibit the proof of concentration rounding. In \pref{sec:anticoncentration} we define anticoncentration and prove that the Gaussian is certifiably anticoncentrated.  Subsequently in \pref{sec:anticoncentration-distributions} we prove that certifiable anticoncentration is closed under linear transformation, and that spherically symmetric strongly log concave distributions are certifiably anticoncentrated.  We defer remaining regression lemmas to the appendix.  In particular, we present our proof of list decoding mean estimation to \pref{sec:mean-estimation}.

\section{Technique Overview: Robust Regression}\label{sec:reg-overview}
In this section we introduce many of the ideas involved in designing our rounding algorithm. We leave sharper error/runtime guarantees and general distributional assumptions to later sections.

Let $N$ be the size of the data set.  Let $\beta < \frac{1}{2}$, and let $M = \beta N$.    
We receive a data set $\cD = \{(X_i,y_i)\}_{i=1}^N$ where the $X_i \in \R^d$ are the covariates and the $y_i \in \R$ are the labels.  Of the covariates, $M$ points are drawn $X_i \sim \cN(0,I)$.  We will refer to these points as "inliers".  Furthermore, we introduce boolean variables $w_1',...,w_N' \in \{0,1\}$ indicating if a data point is an inlier, equal to $1$; outlier, equal to $0$.  Let $\ell' \in \R^d$ be the $\ell_2$ minimizer of the error over the inliers i.e 
$$\ell' \defeq \underset{\ell \in \R^d}{\argmin} \sumn w_i'(y_i - \iprod{\ell,X_i})^2$$
\Pnote{data set must have the clean data and the rest.  Notation is  little inconsistent above,  $\ell'$ must be minimizer on clean data alone} 
Let $\sigma^2$ be a constant upper bounding the $\ell_2$ error  
$$\underset{\ell \in \R^d}{\argmin}  \sumn w_i' (y_i - \iprod{\ell',X_i})^2 \leq \sigma^2$$
Likewise, let $g$ be a constant such that the $\ell_4$ error of $\ell'$ is  
$$\underset{\ell \in \R^d}{\argmin} \sumn w_i'(y_i - \iprod{\ell',X_i})^4 \leq g\sigma^4$$
In the special case of Gaussian noise $\cN(0,\sigma^2)$ we have $g = 3$.  Then for any $d \in \mathbb{Z}^+$ sufficiently large, our algorithm recovers a list of estimates $L = \{\ell_1,...,\ell_A\}$ for $|L| = O\big((\frac{\rho}{\sigma})^{\log(\frac{1}{\beta})}\big)$ such that for some $\ell_i \in L$ we have

$$\norm{\ell_i - \ell'}_2 \leq O\left(\frac{\sigma}{\beta^{3/2}}\right)$$
with high probability over the data. 
If we regard $\frac{\rho}{\sigma}$ as a fixed constant then the list is of length $\poly(\frac{1}{\beta})$.   
Furthermore, our algorithm is \textit{efficient}, running in polynomial time $\poly(d,N)$ where we take $N = d^{O(\frac{1}{\beta^4})}$.  Here we take $N$ to be large enough to certify arbitrary closeness of the $K$'th empirical moments to the distributional moments of the covariates in $\ell_2$ norm for $K = O(\frac{1}{\beta^4})$.  That is to say, for any constant $\epsilon > 0$, and for $X_1,...,X_N \sim \cN(0,I_d)$ we have with high probability

$$
\Norm{\frac{1}{N}\sum_{i=1}^N X_i ^{\otimes \frac{K}{2}}(X_i^{\otimes \frac{K}{2}})^T - M_K}_F^2 \leq \epsilon
$$
where $M_K$ is the empirical $K$'th moment tensor of the spherical Gaussian.  For our analyses we take $d > \rho$ and $\epsilon = O(\sigma d^{-K})$.  For these settings of $N$ and $d$, and fixing $\rho,\sigma, \beta$ to be constants, we often write $o_d(1)$ without ambiguity.

\Pnote{does exponent depend on $\beta$?}

Our approach is to run an SoS SDP, and then \textit{round} out the list.  We begin by describing the Robust Regression SoS SDP.  
\newpage
\begin{algorithm}[H] \label{algo:robustregressionSDP}
\SetAlgoLined
\KwResult{A degree $D$ pseudoexpectation functional $\pE_\zeta$}
 \textbf{Inputs}: $(\cD = \{X_i,y_i\}_{i=1}^N,\rho)$ data set, and upper bound on $\norm{\ell'}$\\
 \Pnote{directly list $\{(X_i,y_i)\}$ as inputs here.  It will be good to make a table like this, as self-contained as possible}
 \begin{eqnarray}
\begin{aligned}
& \underset{\text{degree D pseudoexpectations} \pE}{\text{minimize}}
& & \sum_{i=1}^N \pE[w_i]^2 \\
& \underset{\text{satisfies the polynomial system}}{\text{such that $\pE$}}
&  & (w_i^2 - w_i) = 0, \; i \in [N], \\
& & & \sum_{i=1}^N w_i - M  = 0, \; i \in [N], \\
& & & \Norm{\frac{1}{M}\sum_{i=1}^N w_i X_i ^{\otimes \frac{t}{2}}(X_i^{\otimes \frac{t}{2}})^T - M_t}_F^2 - \epsilon \leq 0  , \;  t \in [K]\\
& & & \frac{1}{M}\sum_{i=1}^N w_i (\langle \ell,X_i\rangle - y_i)^4 - g\sigma^4) \leq 0, \\
& & & \frac{1}{M}\sum_{i=1}^N w_i (\langle \ell,X_i\rangle - y_i)^2 - \sigma^2\leq 0, \\
& & & \Norm{\frac{1}{M}\sum_{i=1}^N w_i (y_i - \iprod{\ell,X_i})X_i}_2^2 = 0, \\
& & & \norm{\ell}^2 - \rho^2 \leq 0\\
\end{aligned}
\end{eqnarray}\\
 \textbf{return}: $\pE_\zeta$
 \caption{RobustRegressionSDP}
\end{algorithm}

\Pnote{do we have to decide on a value of $\epsilon_t$ when we state the theorems?  if so, we can just make it explicit in the SDP.}

\Pnote{it seems like we use $D_4$ to denote $g$ later}

\Pnote{$\hat{\mu}$ in the above SDP is zero, correct?}
Let $\cP$ be the set of polynomial constraints of Robust Regression SDP.  We elaborate on the meaning of each constraint below, and we will often refer to them in our analyses according to the numbering below. 

\begin{enumerate}
    \item The first constraint $\{w_i^2 = w_i\}$ enforces $w_i \in \{0,1\}$ and we refer to it as the \textit{booleaness} constraint.
    \item The next constraint $\{\sum_{i=1}^N w_i - M\}$ ensures we select a $\beta$ fraction of the data set.
    \item The third constraint ensures that the pseudodistribution is over subsets with moments that match the distribution of the covariates. We refer to them as the \textit{moment} constraints.  
    \item The next constraints ensures the error incurred by $\ell$ is small, and we refer to them as the $\ell_2$ \textit{noise} constraint.
    \item Similarly, we have a $\ell_4$ \textit{noise} constraint.  
    \item We have the $\ell_2$ \textit{minimization} constraint, which sets $\ell$ equal to the $\ell_2$ minimizer of the selected $w_i$. \Pnote{$\ell_2$ minimizer of clean data?}
    \item Finally, the \textit{scaling} constraint restricts the length of $\ell$,  $\norm{\ell} \leq \rho^2$.  
\end{enumerate}
The RobustRegression SDP minimizes a convex objective which we refer to as \textit{Frobenius Minimization}.  This technique first used in the work of Hopkins and Steurer \cite{hopkins2017efficient}, ensures that the SDP solution is a convex combination over every possible solution to the system.
This turns out to be crucial.  To see why, consider an actual solution $W_\text{fake}$ consisting of variables $w_1,...,w_N \in \{0,1\}$ and $\ell \in \R^d$ satisfying $\cP$.  The distribution that places mass $1$ on $W_\text{fake}$ and no mass on the clean data is a valid distribution over the solutions to $\cP$ and therefore also a valid pseudodistribution.  Since we only have assumptions on less than half the data, a malicious $W_{fake}$ can be planted anywhere confounding our efforts to recover $\ell'$.  What we need is a way to produce a distribution over solutions to $\cP$ that is a convex combination over all the possible solutions.  The objective function $\sum_{i=1}^N \pE[w_i]^2$ is a strictly convex function, minimizing which ensures that SDP solution is spread over all solutions to $\cP$.  More precisely, we have the following guarantee.

\begin{lemma}\torestate{ \label{lem:FrobeniusMinimization}
(Frobenius minimization$\implies$Correlation)
Let $\cP$ be a polynomial system in variables $\{w_i\}_{i \in [N]}$ and a set of indeterminates $\{\ell_i \}_{i \in \N}$, that contains the set of inequalities:
\begin{align*}
w_i^2 = w_i & \forall i \in [N] & & 
\sum_i w_i = \beta N
\end{align*}
Let $\pE_\zeta: \R[\{w_i\}_{i \in [N]}, \{\ell\}]^{\leq D} \to \R$ denote a degree $D$ pseudoexpectation that satisfies $\cP$ and minimizes the norm $\norm{\pE_\zeta[ w]}$.
If $w'_i \in \{0,1\}$ and $\ell'$ is a satisfying assignment to $\cP$ then there is correlation with the inliers, 
\begin{equation}
\pE_\zeta\left[\frac{1}{M} \sum\limits_{i=1}^N w_iw_i' \right] \geq \beta     
\end{equation}}
\end{lemma}
We defer the proof of this statement to \prettyref{sec:Frobenius}\\
\textbf{Remark}: The lemma does not guarantee that $\pE_\zeta[(\iprod{w,w'} - \beta M)p(w,\ell)] \geq 0$ for all SoS polynomials $p(w,\ell)$ of deg($p(w,\ell)$) $\leq D-1$  .  That is to say, the guarantees of \pref{lem:FrobeniusMinimization} are only in pseudoexpectation.      

At this point we have found a pseudodistribution $\zeta$ satisfying $\cP$ that in pseudoexpectation is correlated with the inliers.  Pursuing a line of wishful thinking, we would like to sample from this pseudodistribution armed with access to its degree $D$ moments.  This is the algorithmic task of rounding the SDP solution, and it is chief intellectual thrust of this work.      

\Pnote{the description below doesn't provide much intuition.  It will be good to use the term: variance.  talk of pseudovariance (probably even define it here), and then use it to state the lemma as follows:
....
For the sake of exposition, let us say that $\pE$ corresponds an actual distribution over solutions to the polynomial system.  
Recall that the goal of the rounding algorithm is to recover the linear function $\ell$.
Suppose the variance of $\ell$ as a random variable is sufficiently small, then its expectation $\pE[\ell]$ is a good approximation to an actual solution.  
.....
}
For the sake of exposition, let us say that $\pE$ corresponds an actual distribution over solutions to the polynomial system.
Recall that the goal of the rounding algorithm is to recover the linear function $\ell'$.
Suppose the variance of $\ell'$ as a random variable is sufficiently small, then its expectation $\pE[\ell]$ is a good approximation to an actual solution.  Formally,

\begin{lemma}
\torestate{\label{lem:regressionSnapping}
Let $\eta \in [0,\frac{1}{2}]$ be a small constant.  
Let $\pE_\zeta$ be a degree $O(\frac{1}{\eta^4})$ pseudoexpectation satisfying the constraints of RobustRegressionSDP($\cD,\rho$).   Then if the pseudovariance of the estimator $\pE_\zeta[\ell]$ is small in every direction
\begin{align}\label{vr}
\max_u \pVar_\zeta[\iprod{\ell,u}] \leq \eta\rho^2 
\end{align}
and there is correlation with the inliers 
\begin{align}\label{fm}
\pE_\zeta\left[ \frac{1}{M}\sum_i w_i w'_i\right] \geq \beta
\end{align}
then our estimator satisfies,
\begin{align}
\norm{\pE[\ell] - \ell'} \leq \sqrt{ \frac{\eta\rho^2 + O(\frac{\sigma^2}{\eta^2})}{\beta}}
\end{align}
In particular, for $\eta = \frac{\beta}{8}$ and $\rho^2 > \Omega(\frac{\sigma^2}{\beta^3})$  the degree $O(\frac{1}{\beta^4})$ pseudoexpectation satisfies 
}
$$
\norm{\pE_\zeta[\ell] - \ell'} \leq \frac{\rho}{2}     
$$
\end{lemma}
\Pnote{using "slash left bracket"  and "slash right bracket", will ensure that the brackets are scaled to the required size, like in second equation above}

\Pnote{degree $1/\eta^4$ seems different from $1/\beta^5$ above?}

Provided we can take the pseudovariance down in every direction, the error guarantee 'contracts' from the trivial $\rho$ to $\frac{\rho}{2}$. It is then possible to iterate such a contraction procedure to achieve optimal error guarantees which is the subject of \pref{sec:iteration}.  
Without going into details, the proof of \pref{lem:regressionSnapping} critically relies on both the concentration and anticoncentration of the covariates.  For instance, if the covariates were drawn from a degenerate point distribution at the origin, then nothing can be inferred about $\ell'$.  In this sense, concentration is insufficient to recover $\ell'$ meaningfully.  To overcome this hurdle, we formalize what it means for a distribution to be SoS certifiably anticoncentrated. 

\textbf{Certifiable Anticoncentration}
As will become clear in \pref{sec:anticoncentration}, the smaller $\eta$ is, the harder it is for SoS to certify the bounds in the above lemma.  For purposes of anticoncentration, $\eta$ is a parameter representing an interval about the origin.  For any distribution $\cD$, we think of $\cD$ as being anticoncentrated if the mass of $\cD$ falling within the $\eta$ interval is small.  For example, in the case of $\cD = \cN(0,1)$, the mass within the $\eta$ interval is upper bounded by $\frac{\eta}{\sqrt{2\pi}}$.  Characterizing this "anticoncentration" of $\cD$ about the origin becomes increasingly difficult (higher degree) as $\eta$ falls, intuitively, because it requires a finer grained picture of the distribution $\cD$.  It turns out the $\cN(0,I_d)$ is $O(\frac{1}{\eta^4})$ SoS certifiably anticoncentrated the proof of which is detailed in \pref{sec:gaussiananticoncentration}.  That this proof is independent of the dimension $d$, along with the rounding algorithm, is what enables the list decoding to run in polynomial time.        
\\
\\
\Pnote{ the above description of anticoncentration seems out of place?  must be moved to a different place, where it is motivated.}

Now we move on to the actual statement of \pref{lem:regressionSnapping}.  In general, the variance of the SDP solution will not be small.  Thus, we will iteratively reduce the variance by conditioning on the $w_i$ variables.  Intuitively, we are conditioning on specific data points $(X_i,y_i)$ being part of the inliers ($w_i = 1$) or being part of the outliers ($w_i = 0$).  

\Pnote{Reader doesn't know what the term "plant" refers to here, and "correlation with the plant".
Only thing the reader knows at this point is what we told him right before lemma 3.2,  "if variance is small then rounding is easy"  and Lemma 3.2 is some quantitative version of that.  We need to take the reader from there, something like:
"The variance of the SDP solution will not be small in general.  We will iteratively reduce the variance by conditioning on the $w_i$ variables.  Intuitively, we are conditioning on specific data points $X_i,y_i$ being part of the clean data"
-- again here "clean data" is a term we should only use if we stated it formally earlier.
}

Towards these ends, let $\pE_{1}, \pE_{2},..., \pE_{R}$ be a sequence of pseudoexpectations where $\pE_1$ is the output of RobustRegressionSDP($\cD,\rho$).  We want to define an algorithm to update $\pE_{t}$ to $\pE_{t+1}$ where $\max_u \pVar_{t+1}[\iprod{\ell,u}] < \max_u \pVar_{t}[\iprod{\ell,u}]$ so that eventually $\pE_{R} < \eta\rho^2$.   

let $\cQ_t$ be the pseudocovariance matrix defined 

$$\cQ_t = \pE_{t}[(\ell - \pE_t[\ell])(\ell - \pE_{t}[\ell])^T]$$
We have $\norm{\cQ_t}_{op} = \max_u \pVar_t[\iprod{\ell,u}]$.  Let's say we have a strategy $\cS$ for selecting a $w_j \in {w_1,...,w_N}$, and then apply the following update 
$$
\pE_{t+1} = \left\{
        \begin{array}{ll}
            \pE_{t}|_{w_j = 1} & \text{with probability} \pE_{t}[w_j]\\
            \pE_{t}|_{w_j = 0} &  \text{with probability} \pE_{t}[1 - w_j]\\
    \end{array}
\right.
$$
Let $b_j$ be a $\{0,1\}$ random variable satisfying $\Pr[b_j = 1] = \pE_t[w_j]$. 
We wish to argue that there is a large expected decrease in the direction of largest variance.   
\begin{align*}
        \norm{\cQ_t}_\text{op}  - \E_{\cS}\E_{b_j}\norm{\cQ_t\big|_{w_j = b_j}}_\text{op} > \text{large}
\end{align*}

\Pnote{good to remind reader what $_{\text{nuc}}$ norm means.  Also, the more standard notation is to use $\norm{Q}_{\infty}$ or $\norm{Q}$ for operator norm, $\norm{Q}_1$ for nuclear norm.}

Unfortunately, controlling $\norm{\cQ}_\text{op}$ is difficult.  We will instead control  
$\norm{\cQ}_\text{nuc}$, i.e trace norm, and prove 

\begin{align*}
        \norm{\cQ}_\text{nuc}  - \E_{\cS}\E_{b_j}\norm{\cQ\big|_{w_j = b_j}}_\text{nuc} \geq
        \Omega(\beta^3\rho^2)
\end{align*}
For the strategy $\cS$ defined below

\begin{lemma} (Variance Decrease Strategy) \torestate{\label{lem:regstrategy}
    Let $\pE_\zeta$ satisfy the pseudoexpectation constraints of RobustRegressionSDP($\cD,\rho$). Let $\cQ$ be the associated pseudocovariance matrix. Let $v$ be any direction in the unit ball $S^{d-1}$.  
    Let $\cS_v$ be a probability distribution over $[N]$\;
  where for any $j \in N$ we have 
  \begin{align*}
      \cS_v(j) \defeq \frac{\pVar_\zeta[w_j(y_j - \pE_\zeta[\langle \ell,X_i \rangle])\langle X_j,v\rangle]}{\sum_{i=1}^N \pVar_\zeta[w_j(y_j - \pE_\zeta[\langle \ell,X_i \rangle])\langle X_j,v\rangle]}
  \end{align*}
  
  Then for $M_4$ being the fourth moment matrix of the Gaussian defined in RobustRegressionSDP,  and for $\norm{\cQ}_{nuc} > \sigma^2\sqrt{g}$ we have   
  
  \begin{align*}
    \pVar_\zeta[\iprod{\ell,v}] -  \E_{j \sim S_v}\E_{b_j} \pVar_\zeta[\iprod{\ell,v}|w_j = b_j]
    \geq  \Omega\Big(\frac{\beta\pVar_\zeta[\iprod{\ell,v}]^2}{\norm{\cQ}_{nuc}}\Big)
\end{align*}
}    
\end{lemma}
\Pnote{Can $\cQ$ be any pseudocovariance matrix above?  or is it the pseudocovariance matrix of $\ell$?  Can we use the notation $\pVar[\ell]$ or something similar to denote pseudocovariance of $\ell$}
\Pnote{$M_4, \epsilon_4$ undefined above.  Also, ``largest eigenpair" does not seem like standard usage}

\Pnote{Does $v$ have to be the largest eigenvector for the above lemma?  If $v$ can be any eigenvector, we should just state for any.   Then, we can explain in text that the lemma lets one decrease the variance along a single direction, thereby decreasing the nuclear norm as follows, and state the corollary}
The above lemma allows the rounding algorithm to decrease the variance along a single direction, thereby decreasing the nuclear norm as follows.  
\begin{corollary}\label{cor:nnr} (Connecting variance decrease strategy to nuclear norm rounding)
For $(\lambda,v)$ being the largest eigenvalue/vector pair of $\cQ$, and $\cS_v$ defined in \pref{lem:regstrategy}. Let $\gamma > 0$ be a constant. If  $\norm{\cQ}_\text{op} \geq \gamma$ and $\norm{\cQ}_{nuc} > \sigma^2\sqrt{g}$, then
\begin{align*}
        \norm{\cQ}_\text{nuc}  - \E_{j \sim S_v}\E_{b_j}\norm{\cQ\big|_{w_j = b_j}}_\text{nuc} \geq \Omega\left(\frac{\beta\gamma^2}{\rho^2}\right)
    \end{align*}  
In particular, for $\gamma = \eta\rho^2$, we have  
 
 \begin{align*}
        \norm{\cQ}_\text{nuc}  - \E_{j \sim \cS_v}\E_{b_j}\norm{\cQ\big|_{w_j = b_j}}_\text{nuc} \geq
        \Omega(\beta\eta^2\rho^2)
    \end{align*}

\end{corollary}
The corollary establishes a win-win.  Either $\norm{\cQ}_{op} < \gamma$ in which case the variance of our estimator is small in every direction, or we can round and decrease an upper bound on $\norm{\cQ}_{op}$.  We defer the proof of \pref{lem:regstrategy} to \pref{sec:reg-rounding} and the proof of \pref{cor:nnr} to section \pref{sec:conditioning}.  

\Pnote{what is the following intuition for?  it appears to be for the entire algorithm and not the above corollary/lemma.  We should have it at a place where we describe the overall strategy, or tell the reader that we are describing the intuition behind the overall algorithm}

Taken together, the conditioning strategy iteratively chases the variance down by selecting the direction of largest variance in the pseudocovariance of our estimator, and conditions on the $w_j$ exhibiting the largest scaled variance.        

We are now ready to state our main algorithm and prove the main result of this section 

\Pnote{what is $\ell'$ in the theorem statement below? we should say, for any integral solution $\ell'$, we recover something close to $\ell'$ with probability $\beta/160$}
\begin{algorithm}[H] \label{algo:regressionroundtree}
\SetAlgoLined
\KwResult{a $d$ dimensional hyperplane}
 \textbf{Inputs}: $(\pE_1,\rho)$ The output of RobustRegressionSDP and the scaling parameter\\
 \For{$t = 1:O(\frac{1}{\eta^2\beta^2})$}{
  Let $\cQ_t = \pE_t[(\ell - \pE_t[\ell])(\ell - \pE_t[\ell])^T]$ be the pseudocovariance matrix of the estimator $\ell$\\
 \text{Let }$(\lambda,v)$\text{ be top eigenvalue/vector of }$\cQ_t$ \\
 \eIf{ $\lambda > \eta\rho^2$}{
  Let $\cS_v$ be a probability distribution over $[N]$\;
  Where for any $j \in N$ we have $\cS_v(j) \defeq \frac{\pVar_t[w_j(y_j - \pE_t[\langle \ell,X_i \rangle])\langle X_j,v\rangle]}{\sum_{i=1}^N \pVar_t[w_j(y_j - \pE_h[\langle \ell,X_i \rangle])\langle X_j,v\rangle]}$ \\
  Sample $j \sim \cS_v$ \\
  Sample $b_j \in $ Bern($\pE_t[w_j]$)\\
  Let $\pE_{t+1} = \pE_t\big|_{w_j = b_j}$\\
 }
 {
 \textbf{return:} $\pE_t[\ell]$
 }
 }
 \caption{Regression Rounding Algorithm}
\end{algorithm}

\begin{theorem} \label{thm:main-regression-estimation}
    Let $\ell'$ be a solution to the constraints of RobustRegressionSDP.  Let $\eta$ be a constant greater than $0$.  Let $\pE_{\zeta}$ be the output of $\text{RobustRegressionSDP}(\cD,\rho)$ for degree $D = O\Big(\max(\frac{1}{\beta^2\eta^2},\frac{1}{\eta^4})\Big) $.  Then after $R = O(\frac{1}{\beta^2\eta^2})$ rounds of updates according to the strategy $\cS$ in \pref{lem:regstrategy}, the resulting pseudoexpectation, which we denote $\pE_R$, satisfies
    $$ \norm{\pE_{R}[\ell] - \ell'} \leq \sqrt{\frac{\eta\rho^2 + O\big(\frac{\sigma^2}{\eta^2})}{\beta}}$$
    with probability greater than $\Omega(\beta)$ over the randomness in the algorithm.  In particular for $\eta = \Omega(\beta)$ and for $\rho^2 \geq \Omega\big(\frac{\sigma^2}{\beta^3} \big)$, the degree $D = O(\frac{1}{\beta^4})$ pseudoexpectation satisfies  
    $$ \norm{\pE_{R}[\ell] - \ell'} \leq \frac{\rho}{2}$$
\end{theorem}

\textbf{Remark}:    
\Pnote{technically, this is a list decoding algorithm already, since one gets the full list by running the algorithm many many times, and boosting the probability that something close to $\ell'$ appears.  Alternately, one can think of going over all possibilities for the random choices made by the algorithm}

As stated, \pref{thm:main-regression-estimation} takes down the error guarantee to $\frac{\rho}{2}$ and is not yet an iterative algorithm that obtains the optimal error guarantees, yet it contains most of the elements of the full algorithm.  Issues concerning iteration are the subject of the next section on algorithms.          

\begin{proof}
We now have all the tools to prove \pref{thm:main-regression-estimation}.  By frobenius minimization \pref{lem:FrobeniusMinimization} we have, 
$$\pE_{\zeta}[\sumn w_iw_i'] \geq \beta$$ 
Now we show that after $R$ rounds of conditioning,
$$ \max_{u \in \cS^{d-1}} \pVar_{\zeta,R}[\iprod{\ell,u}] \leq  \eta\rho^2$$

To apply \pref{lem:regressionSnapping} we iteratively round according to $\cS_v$ in  \pref{lem:regstrategy} to decrease $\norm{\cQ}_{nuc}$.  For $\norm{\cQ}_{op} > \eta\rho^2$, \pref{cor:nnr} gives us 

\begin{align*}
        \norm{\cQ}_\text{nuc}  - \E_{j \sim S_v}\E_{b_j}\norm{\cQ\big|_{w_j = b_j}}_\text{nuc} \geq \Omega(\beta\eta^2\rho^2)
\end{align*}  

\Pnote{too many sentences starting with words like Or, Either, Where,  perhaps we should reduce that a bit}

We aim to show that after $R = O(\frac{1}{\beta^2\eta^2})$ iterations, the algorithm outputs $\norm{\cQ_R}_\text{nuc} \leq \eta\rho^2$ with probability greater than $1 - \frac{\beta}{4}$ over the randomness in the selection strategy and ${0,1}$ conditionings.  We denote the probability and expectation over the randomness in the algorithm $\Pr_\cA[\cdot]$ and $\E_\cA[\cdot]$ respectively.  Thus, to prove the following  
$$\Pr_{\cA}[\norm{\cQ_R}_{op} \leq \eta\rho^2] \geq 1 - \frac{\beta}{4}$$
we proceed by contradiction.  Suppose that at each iteration of $t = 1,2,...,R$, that $\norm{\cQ_t}_\text{op} > \eta\rho^2$ with probability greater than $\frac{\beta}{4}$.  Then in expectation over $\cS$ we have that each round of conditioning decreases $\norm{\cQ}_\text{nuc}$ by $\frac{\beta}{4}$ (the probability that the assumption in \ref{eq:overview1} holds) times $\Omega(\beta\eta^2\rho^2)$ (the expected decrease in \ref{eq:overview1}).  Thus, 

\begin{align}\label{eq:overview1}
\E_{\cA}[\norm{\cQ_t}_\text{nuc}] - \E_{\cA}[\norm{\cQ_{t+1}}_\text{nuc}] \geq \Pr_{\cA}[\norm{\cQ_t}_\text{op} \geq \eta\rho^2]\cdot \Omega(\beta\eta^2\rho^2) \geq \Omega(\beta^2\eta^2\rho^2)
\end{align}

We also know that the initial pseudocovariance is upper bounded in nuclear norm i.e  

\begin{align}\label{eq:overview2}
\norm{\cQ_{\zeta}}_\text{nuc} = \pE_\zeta[\norm{\ell - \pE_\zeta[\ell]}^2] = \pE_\zeta[\norm{\ell}^2] - \pE_\zeta[\norm{\ell}]^2 \leq \pE_\zeta[\norm{\ell}^2] \leq \rho^2
\end{align}

Where the last inequality is an application of the scaling constraint (7).  Thus, putting together \ref{eq:overview1} and \ref{eq:overview2} after $R = O(\frac{\rho^2}{\beta^2\eta^2\rho^2}) = O(\frac{1}{\beta^2\eta^2})$ iterations,   $\E_{\cA}[\norm{\cQ_R}_\text{nuc}] \leq 0$ which is impossible because $\norm{\cQ_R}_\text{nuc} = \pE_R[\norm{\ell - \pE_R[\ell]}^2] \geq 0$.  Thus our assumption is false, and   $\Pr_{\cA}[\norm{\cQ_R}_{op} \leq \eta\rho^2] \geq 1 - \frac{\beta}{4}$ as desired.

We also know by the law of total pseudoexpectation that in expectation over the selection strategy and ${0,1}$ conditionings,    
\[ \E_{\cA}\left[\pE_{R}\left[\sumn w_iw_i'\right]\right]= \pE_{\zeta}\left[\sumn w_iw_i'\right] \geq \beta \]
\Pnote{ $slash left bracket$  and $slash right bracket$} Note that this is a generic fact that is true regardless of which conditioning strategy we choose.  Thus by Markov for random variables taking values in $[0,1]$ we have 

$$\Pr_{\cA}\left[\pE_{R}\left[\sumn w_iw_i'\right] \geq \frac{\beta}{2}\right] \geq \frac{\beta}{2}$$
Now that we know $\Pr_{\cA}[\norm{\cQ_R}_{op} \leq \eta\rho^2] \geq 1 - \frac{\beta}{4}$ and $\Pr_{\cA}\left[\pE_{R}\left[\sumn w_iw_i'\right] \geq \frac{\beta}{2}\right] \geq \frac{\beta}{2}$, we conclude via union bound that the failure probability of both events is upper bounded by $\frac{\beta}{4} + 1 - \frac{\beta}{2} = 1 - \frac{\beta}{4}$.  Thus  
the conditions of \pref{lem:regressionSnapping} are satisfied with probability greater than $\frac{\beta}{4}$ in which case 
$$ \norm{\pE_{R}[\ell] - \ell'} \leq \sqrt{\frac{\eta\rho^2 + O\big(\frac{\sigma^2}{\eta^2})}{\beta}}$$
     In particular for $\eta = \Omega(\beta)$ and for $\rho^2 \geq \Omega\big(\frac{\sigma^2}{\beta^3} \big)$ we have  
     $$ \norm{\pE_{R}[\ell] - \ell'} \leq \frac{\rho}{2}$$

\Pnote{we should conclude, just to remind the reader why this is desired?}
\end{proof}
\begin{lemma}\torestate{\label{lem:roundtree}
Running Algorithm \pref{algo:regressionroundtree} a total of $O(\frac{1}{\beta})$ times produces a list $L = \{\ell_1,...,\ell_{O(\frac{1}{\beta})}\}$ such that with probability 1 - c, there exists a list element $\ell_i \in L$ satisfying 
$\norm{\ell_i - \ell'} \leq \frac{\rho}{2}$
where $c$ is a small constant. Minor modifications enable the algorithm to succeed with high probability. } 
\end{lemma}
We defer the modifications required to succeed with high probability to the appendix. We proceed under the assumption that \pref{algo:regressionroundtree} outputs a list $L$ satisfying the guarantees in \pref{lem:roundtree} with high probability. For variety, we present the mean estimation algorithms with these modifications in place.   
\Pnote{ $\eta$ needs to be quantified in above lemma?}

\section{Iterative Contraction for Sharp Rates}\label{sec:iteration}

\Pnote{different title for the section?}

\Pnote{describe at a high level, what is happening in this section?  
The RoundTree algorithm in previous section gets approximate.  Here we boost it as follows"}
The Regression Rounding  \pref{algo:regressionroundtree} generates a list $L$ which contracts the error guarantee from $\rho$ to $\frac{\rho}{2}$ with high probability.  In this section we wrap the algorithm in an iterative framework to obtain sharp error guarantees.

\Pnote{the following remark on parameter setting is a bit tedious on the reader.  The error terms in the above lemma are also messy.  We should fix the correct choice of $\eta$ in the lemma statement itself, and only state the lemma with the simple error term   $\rho/2$}
\Pnote{ above discussion is too technical for a reader at this stage.  At this stage, the reader doesn't know the overall scheme of things yet.  We should completely avoid the discussion on parameter settings, and instead use the text to explain the general schema of things more carefully.}

\Pnote{the following two algorithm descriptions are really difficult for the reader to understand, without any accompanying explanation.  Ideally, the accompanying explanation is good enough that the reader skips reading the pseudocode altogether}

Our iterative framework, ListDecodeRegression \pref{algo:regressioniterative}, iterates over the list $L$ generated by RoundingRobustRegression \pref{algo:regressionroundtree}, and uses the list elements to shift the data so as to obtain sharper estimates.  This will involve rerunning both RobustRegressionSDP \pref{algo:robustregressionSDP} and RoundingRobustRegression \pref{algo:regressionroundtree}. Formally, for each $\ell_i \in L$ create a new dataset $\{(X_j,y_j')\}_{j=1}^N$ with the same covariates with shifted labels $y_1',...,y_N'$.   The labels are shifted according to the hyperplane $\ell_i$ as follows, $y_j' := y_j - \iprod{\ell_i,X_j}$ for all $j \in [N]$. Then the scaling constraint $\{\norm{\ell}^2 \leq \frac{\rho^2}{2} \}$ is added to the RobustRegressionSDP, and we resolve the SDP and rerun the rounding.  Iterating this procedure, the error guarantee contracts each iteration from $\frac{\rho}{2}, \frac{\rho}{4}, ...$  so on and so forth, whilst the list length increases multiplicatively by factors of $O(\frac{1}{\beta})$ until the ubiquitous assumption $\rho^2 \geq \Omega(\frac{\sigma^2}{\beta^3})$ no longer holds and we are left with the sharp error guarantee $\norm{\ell_i - \ell'} \leq O(\frac{\sigma}{\beta^{3/2}})$ for some list element $\ell_i \in L$.  The following theorem formalizes the discussion above.

\begin{algorithm}[H] \label{algo:regressioniterative}
\SetAlgoLined
\KwResult{A list of hyperplanes $L = \{\ell_1,...,\ell_A\}$}
 \textbf{inputs}: $(\cD,\rho)$\\
 $L = \{0\}$\\
 \For{ $t \in \log_2(\frac{\rho\beta^{3/2}}{\sigma})$}{
    $\rho_t = \frac{\rho}{2^t}$\\
    \% Let $Y$ to be a list of pseudoexpectations\\
    $Y = \emptyset$ \\
    \For{$\ell_i \in L$}{
        $(\cX,\cY) = \cD$\\
        \For{$y_j \in \cY$}{$y_j = y_j - \iprod{\ell_i,X_i}$ }
        Let $Y = Y \cup \text{RobustRegressionSDP}(\cD,\rho_t)$\\
    }
    $L = \emptyset$\\
    \For{$\pE_\zeta \in Y$}{
        $L' = \text{RegressionRounding}(\pE_\zeta,p_t)$\\
        $L = L \cup L'$\\
    }
 }
 \textbf{return:} L
 \caption{ListDecodeRegression}
\end{algorithm}

\Pnote{In the theorem below, it will be good to remind the reader what $\ell'$ is, it is any solution to the polynomial system}

\begin{theorem}\label{thm:reg-final}
ListDecodeRegression($\cD,\rho$) outputs a list of hyperplanes $L = \{\ell_1,...,\ell_A\}$ where  
$A = O\big((\frac{1}{\beta})^{\log(\frac{\beta^{3/2}\rho}{\sigma})}\big)$ such that for some $\ell_i \in L$ 

$$\norm{\ell_i - \ell'} \leq O\Big(\frac{\sigma}{\beta^{3/2}}\Big)$$
with high probability in time $\big(\frac{\rho}{\sigma}\big)^{\log(1/\beta)} N^{O(\frac{1}{\beta^4})}$ for $N = d^{O(\frac{1}{\beta^4})}$.  Here we are running solving RobustRegressionSDP \pref{algo:robustregressionSDP} for degree $D = O(\frac{1}{\beta^4})$, and running $R = O(\frac{1}{\beta^4})$ rounds of the RegressionRounding  \pref{algo:regressionroundtree} 
\end{theorem}
\begin{proof}
For any call to RegressionRounding  \pref{algo:regressionroundtree}, we have by \pref{cor:roundtree} a list $L$ and a list element $\ell_i \in L$ satisfying  $\norm{\ell_i - \ell'} \leq \frac{\rho}{2}$.  After each iteration we construct a new data set $\{(X_j,y_j')\}_{j=1}^N$ by shifting the labels according to the rule $y_j' := y_j - \iprod{\ell_i,X_j}$ and enforce the scaling constraint $\{\norm{\ell}^2 \leq \frac{\rho^2}{4}\}$.  The key point is that this new constraint is feasible for at least one iterate $\ell_i \in L$.  This is all we need to iterate RobustRegressionSDP \pref{algo:robustregressionSDP} and subsequently the RegressionRounding Algorithm \pref{algo:regressionroundtree}.  

The list length grows by a factor of $O(\frac{1}{\beta})$ per iteration for $O(\log(\frac{\beta^{3/2}\rho}{\sigma}))$ iterations.  Thus, we run RobustRegressionSDP \pref{algo:robustregressionSDP} no more than  $O\big((\frac{1}{\beta})^{\log(\frac{\beta^{3/2}\rho}{\sigma})}\big)$ times.  From \pref{lem:roundtree} we solve RobustRegressionSDP for degree $D=O(\frac{1}{\beta^4})$. This concludes our treatment of list decoding robust regression.  
\end{proof}
Thus far we have assumed the covariates are distributed $\cN(0,I)$ with a fourth injective tensor norm of $B = 3$. In addition, we regarded the fourth moment of the noise model upper bounded by $g\sigma^4$ for a constant $g$. We conclude this section by stating a general theorem relevant for large values of $B$ and $g$.  

\begin{theorem}
Let a $\beta$ fraction of $X_1,...,X_N \in \R^d$ be drawn from a distribution $\cD$ with identity covariance and a fourth injective tensor norm upper bounded by a constant $B$.  Let $N,d,\ell',g,\sigma,\rho$ be defined as they were previously.  Then  ListDecodeRegression($\cD,\rho$) outputs a list of hyperplanes $L = \{\ell_1,...,\ell_A\}$ where  
$A = O\big((\frac{1}{\beta})^{\log(\frac{\beta^{3/2}\rho}{\sigma})}\big)$ such that for some $\ell_i \in L$ 

$$\norm{\ell_i - \ell'} \leq O\Big(\max\big(\frac{\sigma}{\beta^{3/2}},\sigma^2\sqrt{g}\big)\Big)$$
with high probability in time $\big(\frac{\rho}{\sigma}\big)^{\log(1/\beta)} N^{O(\frac{B}{\beta^4})}$ for $N = d^{O(\frac{1}{\beta^4})}$.  Here we are running solving RobustRegressionSDP \pref{algo:robustregressionSDP} for degree $D = O(\frac{B}{\beta^4})$, and running $R = O(\frac{B}{\beta^4})$ rounds of the RegressionRounding  \pref{algo:regressionroundtree}
\end{theorem}
The proof follows by direct inspection of the proof of \pref{thm:reg-final}.  

\section{On Conditioning SoS SDP Solutions}\label{sec:conditioning}
In this section we prove facts about concentration rounding.   
\subsection{Concentration Rounding: One Dimensional Case}

\begin{fact} (Conditional Pseudoexpectation Toolkit)

For any two polynomials $p,q$ we define $\pVar[p] \defeq \pE[p^2] - \pE[p]^2$, $\pCov[p,q] \defeq \pE[p,q] - \pE[p]\pE[q]$. 

Let $w$ satisfy the boolean constraint $w^2 = w$.  For a random variable $b$ taking values in $\{0,1\}$ such that $\Pr[b = 1] = \pE[w]$ and $\Pr[b = 0] = \pE[1 - w]$ we have the following useful facts

\begin{enumerate}
    \item (Law of Total Pseudoexpectation) $\E_b\pE[p(x,w)|w = b] = \pE[p(x,w)]$
    \item (Law of Total Pseudovariance) $\pVar[p(x)] - \E_{b}\pVar[p(x,w)|w = b] = \Var_b[\pE[p(x,w)|w = b]]$\\
    \item (Correlation Identity) $\pE[p(x,w)|w = b] = \frac{\pCov[p(x,w),w]}{\pVar[w]}b + (\pE[p(x,w)] - \frac{\pCov[p(x,w),w]}{\pVar[w]}\pE[w])$
\end{enumerate}

\end{fact}
    
\begin{proof} (facts)  It is easy to check that $\E[b] = \pE[w]$ and $\Var[b] = \pVar[w]$ and $\Cov[b] = \pCov[w]$.  The law of total pseudoexpectation is an application of definitions.  The law of total pseudovariance is an application of the law of total pseudoexpectation. The proof is as follows.   
$$\pVar[p(x,w)] -  \E_{b}\pVar[p(x,w)|w = b] = \pE[p(x,w)^2] - \pE[p(x,w)]^2 - (\E_b\pE[p(x,w)^2] - \E_b\pE[p(x,w)]^2)$$
$$= \E_b[\pE[p(x,w)]^2] - \pE[p(x,w)]^2 = \E_b[\pE[p(x,w)]^2] - \E_b[\pE[p(x,w)|b]]^2 = \Var_b[\pE[p(x,w)|w = b]]$$

Lastly, we prove the correlation identity.  We know $\pE[p(x,w)|w = b]$ is a function of $b$.  Therefore there exists constants $c$ and $d$ such that $\pE[p(x,w)|w = b] = cb + d$.  First we determine $c$.  We know 
\begin{align*}
c\pVar[w] = c\Var[b] = \Cov(cb + d,b) = \Cov(\pE[p(x,w)|w=b],b)    
\end{align*}
\begin{align*}
    = \E[b\pE[p(x,w)|w=b]] - \E[\pE[p(x,w)|w=b]]\E[b]
\end{align*}
\begin{align*}
    = \Pr[b=1]\pE[p(x,w)|w=1] - \pE[p(x,w)]\E[b]
\end{align*}
\begin{align*}
    = \pE[w]\frac{\pE[p(x,w)w]}{\pE[w]} - \pE[p(x,w)]\pE[w]
\end{align*}
\begin{align*}
    = \pCov[p(x,w),w]
\end{align*}

\Pnote{the above set of equations must be aligned, using say $begin\{align\}$}

Thus 
$$c = \frac{\pCov[p(x,w),w]}{\pVar[w]} $$
Then to obtain $d$ we apply expectation on both sides of $\pE[p(x,w)|w = b] = cb + d$.   
\end{proof}
Let $w_1,\ldots, w_N$ be variables that satisfy the boolean constraint for all $i \in [N]$. 
Let $Z_1,\ldots, Z_N \in \R$ be numbers and let $\hu \defeq \frac{1}{N} \sum\limits_{i=1}^N w_i Z_i$.
We show that pseudo-variance $\pVar[\hu]$ decreases in expectation when we condition on the variables $w_i$ according to a carefully chosen strategy.

\begin{theorem}\label{thm:onedimrounding}
    Let $w_1,\ldots, w_N$ denote variables satisfying $\{ w_i^2 = w_i | i \in [N]\}$ and let $\hu = \frac{1}{N}\sum_{i \in [N]} w_i Z_i$ for some sequence of real numbers $\{Z_i\}_{i \in [N]}$.
    Define a probability distribution $\cS: [N] \to \R^+$ as
$$\cS(j) \defeq \frac{\pVar[ w_jZ_j]}{\sum_{i=1}^N \pVar[w_jZ_j]} $$
    If we condition on the value of $w_j$ where $j$ is drawn from $\cS$,  then the pseudovariance decreases by
    \begin{align*}
\pVar[\hu] - \E_{j \sim \cS}\E_{b_j} \pVar[\hu|w_j = b_j] \geq \frac{ \big(\pVar(\hu)\big)^2}{\frac{1}{ N}\sum\limits_{i=1}^N \pVar(z_i)}
    \end{align*}
    
    Where $b_j$ is $[0,1]$ random variable with $\Pr[b_j = 1] = \pE[w_j]$ and $\Pr[b_j = 0] = \pE[1 - w_j]$
    
    This also immediately yields for $\hu = \sumn w_iZ_i$  
    
    \begin{align*}
\pVar[\hu] - \E_{j \sim \cS}\E_{b_j} \pVar[\hu|w_j = b_j] \geq \beta\frac{ \big(\pVar(\hu)\big)^2}{\sumn \pVar(z_i)}
    \end{align*}
\end{theorem}

\begin{proof}
Let $z_i = w_iZ_i$ for all $i \in [N]$.  
Since $z_i$ is a constant multiple of $w_i$, conditioning on $z_i$ is equivalent to conditioning on $w_i$. 
We begin with the law of total variance
\begin{align*}
    \pVar(\hu) - \E_{b_j} \pVar(\hu|w_j = b_j) =  \Var_{b_j}(\pE[\hu|w_j = b_j])
\end{align*}

Then we apply the expectation over the strategy $\cS$ to both sides to obtain 
\begin{align*}
    \pVar[\hu] - \E_\cS \E_{b_j} \pVar[\hu|w_j = b_j] =  \E_{\cS}\Var_{b_j}[\pE[\hu|w_j = b_j]]
\end{align*}

\begin{align*}
= \pE_{\cS}\pVar_{b_j}\Big[\frac{\pCov[\hu,w_j]}{\pVar[w_j]}b_j\Big] = \pE_{\cS}\frac{\pCov[\hu,w_j]^2}{\pVar[w_j]^2}\Var[b_j] = \pE_{\cS}\frac{\pCov[\hu,z_j]^2}{\pVar[z_j]}
\end{align*}
\Pnote{above set of equations needs to be aligned}
Writing out the distribution of $\cS$ we obtain 
$$ 
= \sum_{j=1}^N \frac{\pVar[z_j]}{\sum_{i=1}^N \pVar[z_j]}\frac{\pCov[\hu,z_j]^2}{\pVar[z_j]} =  \frac{\sum_{i=1}^N\pCov[\hu,z_j]^2}{\sum_{i=1}^N \pVar[z_j]} =  \frac{\frac{1}{N}\sum_{i=1}^N\pCov[\hu,z_j]^2}{\frac{1}{N}\sum_{i=1}^N \pVar[z_j]}
$$
by Jensen's inequality 
$$ 
\geq \frac{(\frac{1}{N}\sum_{i=1}^N\pCov[\hu,z_j])^2}{\frac{1}{N
}\sum_{i=1}^N \pVar[z_j]}
 = \frac{\pVar[\hu]^2}{\frac{1}{N}\sum_{i=1}^N \pVar[z_i]}$$

\end{proof}

\Pnote{ we should define a macro for $_\text{nuc}$ and $_\text{op}$ to make switching notation easier}

\begin{corollary} (Connecting variance decrease strategy to nuclear norm rounding)
  For $\cS_v$ and $\cQ$ defined in \pref{lem:regstrategy}. Let $\gamma > 0$ be a constant.  If  $\norm{\cQ}_\text{op} \geq \gamma$ and $\norm{\cQ}_{nuc} > \sigma^2\sqrt{g}$, then
\begin{align*}
        \norm{\cQ}_\text{nuc}  - \E_{j \sim S_v}\E_{b_j}\norm{\cQ\big|_{w_j = b_j}}_\text{nuc} \geq \Omega(\frac{\beta\gamma^2}{\rho^2})
\end{align*}  
 In particular for $\gamma = \eta\rho^2$, we have  
 
 \begin{align*}
        \norm{\cQ}_\text{nuc}  - \E_{j \sim \cS_v}\E_{b_j}\norm{\cQ\big|_{w_j = b_j}}_\text{nuc} \geq
        \Omega(\beta\eta^2\rho^2)
    \end{align*}
 
\end{corollary}

\begin{proof}
Let $v, e_1,...,e_{d-1} \in R^d$ be an orthonormal basis.  First we write the nuclear norm as a decomposition along an orthonormal basis i.e

\begin{align*}
\norm{\cQ}_\text{nuc} =   \pVar[\langle \ell, v\rangle]  +  \sum\limits_{j=1}^{d-1} \pVar[\langle \ell, e_j\rangle]
\end{align*}
Now we write down the expected decrease in $\norm{\cQ}_\text{nuc}$ for a single conditioning to obtain
\begin{align*}
    \norm{\cQ}_\text{nuc} - \E_{j \sim \cS_v}\E_{b_j}\norm{\cQ\big|_{w_i = z_i}}_\text{nuc} =  \big(\pVar[\langle \ell, v\rangle]  - \E_{j \sim S_v}\E_{b_j}\pVar[\langle \ell, v\rangle|w_j = b_j]\big)\\  +  \big(\sum\limits_{j=1}^{d-1} \pVar[\langle \ell, e_j\rangle] - \E_{j \sim S_v}\E_{b_j}\pVar[\langle \ell, e_j\rangle|w_j = b_j]\big)
\end{align*}
  Then we apply \pref{lem:regstrategy} to the first term, and we apply the fact that pseudovariance is monotonically decreasing after conditioning (law of total pseudovariance) to the second term to obtain.
\Pnote{perhaps this part of the calculation should appear elsewhere, may be the entire lemma?  because the particular error term from Lemma 3.3 looks rather complicated, and out of place in this section}

\begin{align*}
    \geq   \Omega(\frac{\beta\pVar_\zeta(\iprod{\ell,v})^2}{\norm{\cQ}_{nuc}}) - o_d(1)
\end{align*}
Using the fact that $\pVar_\zeta(\iprod{\ell,v}) = \norm{\cQ}_\text{op} \geq \gamma$ and $ \norm{\cQ}_\text{nuc} = \pE_\zeta[\norm{\ell - \pE_\zeta[\ell] }^2] \leq \pE_\zeta[\norm{\ell}^2]\leq \rho^2$ we further lower bound by  
\begin{align*}
    \geq \Omega(\frac{\beta\gamma^2}{\rho^2}) - o_d(1) 
    \geq \Omega(\frac{\beta\gamma^2}{\rho^2}) - o_d(1) 
\end{align*}
for $\gamma = \eta\rho^2$, we conclude

\begin{align*}
    \norm{\cQ}_\text{nuc} - \E_{j \sim \cS_v}\E_{b_j}\norm{\cQ\big|_{w_i = z_i}}_\text{nuc} \geq  \Omega(\beta\eta^2\rho^2) 
\end{align*}
as desired. 
\end{proof}

\section{Frobenius Minimization}

\subsection{Frobenius Norm Minimization} \label{sec:Frobenius}

    \restatelemma{lem:FrobeniusMinimization}

\begin{proof}
Let $\pE_P$ denote the pseudo-expectation operator corresponding to the actual assignment $\{w'_i\}_{i \in [N]}$ and $\{\ell'\}$.
Note that $\pE_P$ is an actual expectation over an assignment satisfying the polynomial constraints.
For a constant $\kappa \in [0,1]$, let us consider the pseudoexpectation operator $\pE_R$ defined as follows for a polynomial $p(w)$,
$$\pE_R  \defeq  \kappa \pE_P + (1-\kappa)\pE_D$$
Since $\pE_D$ is the pseudoexpecation operator that minimizes $\norm{\pE_D[w]}$, we get that
\begin{equation}\label{eqfm1}
\iprod{\pE_R[w], \pE_R[w]} \geq \iprod{\pE_D[w], \pE_D[w]}
\end{equation}

Expanding the LHS with the definition of $R$ we have
\begin{align*}
(1-\kappa)^2 \cdot \iprod{\pE_D[w], \pE_D[w]} + 2 \kappa (1-\kappa) \iprod{\pE_D[w], \pE_P[w]} + \kappa^2 \iprod{\pE_P[w],\pE_P[w]} \geq \iprod{\pE_D[w], \pE_D[w]}
\end{align*}
Rearranging the terms we get
\begin{align*}
  \iprod{\pE_D[w], \pE_P[w]}  \geq \frac{1}{2 \kappa(1-\kappa)} \left((2\kappa - \kappa^2)\iprod{\pE_D[w], \pE_D[w]} -\kappa^2 \iprod{\pE_P[w],\pE_P[w]}\right)
\end{align*}
By definition, we have that $\iprod{\pE_P[w], \pE_P[w]} = \sum_i w_i^{'2} = \beta N$.  By Cauchy-Schwartz inequality, $\iprod{\pE_D[w],\pE_D[w]} \geq \frac{1}{N} \left(\sum_i \pE_D[w_i]\right)^2 = \frac{1}{N} (\beta N)^2 = \beta^2 N $.  Substituting these bounds back we get that,
\begin{align*}
    \iprod{\pE_D[w], \pE_P[w]} \geq \frac{\left((2\kappa - \kappa^2)\beta^2 -\kappa^2 \beta\right)}{2 \kappa(1-\kappa)}  \cdot N 
\end{align*}
Taking limits as $\kappa \to 0$, we get the desired result.

\end{proof}

\section{Regression Rounding}\label{sec:reg-rounding}
In this section we prove that concentration rounding decreases $\norm{\cQ}_{nuc}$.  First we closely approximate $\pVar[\iprod{\ell,u}]$ by $\pVar[\sumn w_i(y_i - \langle \pE[\ell],X_i\rangle)\langle X_i, u \rangle]$ for any unit vector $u \in S^{d-1}$.  Then we apply \pref{thm:onedimrounding} with the strategy $\cS_v$ to analyze a single iteration of concentration rounding.     
We begin with the following useful lemma for working with pseudovariance.  
\Pnote{we should state what we are trying to do in this section, the reader doesn't know what to expect in this section}
\begin{lemma}[Pseudovariance Triangle Inequality]
Let $f(x)$ and $g(x)$ be polynomials. Then for any $\psi > 0$ there is a degree $2$ SoS proof of the following.  
\begin{equation*}
    \pVar[f(x) + g(x)] \leq (1+\psi)\pVar[f(x)] + (\frac{1 + \psi}{\psi})\pVar[g(x)]
\end{equation*}
\end{lemma}
\Pnote{$\psi$ is not quantified here}
\begin{proof}
\begin{equation}\label{psi1}
    \pVar[f(x) + g(x)] = \pE[(f(x)+ g(x) - \pE[f(x + g(x))])^2] = \pE[((f(x) - \pE[f(x)]) + (g(x) - \pE[g(x)]))^2]
\end{equation}
Then we observe that there is a degree $2$ SoS proof of the fact
\begin{align*}
    (f(x) + g(x))^2 = f(x)^2 + g(x)^2 + 2\psi f(x)\frac{g(x)}{\psi} \leq (1 + \psi^2)f(x)^2 + (\frac{1 + \psi^2}{\psi^2})g(x)^2 
\end{align*}
Plugging this into \ref{psi1}
\begin{align*}
    \leq (1 + \psi^2)\pE[(f(x) - \pE[f(x)])^2] + (\frac{1 + \psi^2}{\psi^2})\pE[(g(x) - \pE[g(x)])^2] = (1 + \psi^2)\pVar[f(x)] + \frac{1 + \psi^2}{\psi^2}\pVar[g(x)]
\end{align*}
Substituting any variable $\psi' = \psi^2 > 0$ we obtain the desired result.  
\end{proof}

\restatelemma{lem:regstrategy}

\Pnote{the error term looks rather messy.  what is $\epsilon_4$, and can we substitute the value for $B$.  Also, is the final term $\pE_{\zeta}[ \norm{\ell - \pE[\ell]}^2]$ equal to trace of pseudocovariance matrix of $\ell$.  If so, we should just define the pseudocovariance matrix at some point, and use it in this statement}

To prove \pref{lem:regstrategy} we will need the following lemma 

\begin{lemma}\torestate{\label{lem:reg-approx-round}
Let $\pE$ be a pseudodistribution satisfying $\cP$.  The following holds.
\begin{equation}\label{rr5}
    \pVar[\langle \ell,u\rangle] \leq (1 + o_d(1))\pVar[\sumn w_i(y_i - \langle \pE[\ell],X_i\rangle)\langle X_i, u \rangle] +  o_d(1) 
\end{equation}

\begin{equation}\label{rr6}
 \pVar[\sumn w_i(y_i - \langle \pE[\ell],X_i\rangle)\langle X_i, u \rangle]\leq (1 + o_d(1))\pVar[\langle \ell,u\rangle] + o_d(1)
\end{equation}
}
\end{lemma}

\Pnote{Is there a with high probability in the above statement?  what is the error $1/poly(d)$ coming from?}
Informally, \ref{lem:reg-approx-round} gives us an arbitrarily good approximation (up to a negligible additive error term) to $\pVar[\langle \ell,u\rangle]$ by the variance of an estimator that is amenable to rounding via \pref{thm:onedimrounding}.  We defer the proof to the end of the section.  Now we're ready to prove \pref{lem:regstrategy}
\begin{proof} (Proof of \pref{lem:regstrategy})
First we apply \pref{lem:reg-approx-round} to obtain an arbitrarily good constant factor approximation of the variance decrease.
\begin{align*}
  \pVar_\zeta[\iprod{\ell,v}] -  \E_{j \sim S_v}\E_{b_j} \pVar_\zeta[\iprod{\ell,v}|w_j = b_j]  
\geq (1 - o_d(1))\pVar\left[\sumn w_i(y_i - \langle \pE[\ell],X_i\rangle)\langle X_i,u\rangle\right] \\- (1+o_d(1)) \E_{j \sim S_v}\E_{b_j} \pVar\left[\sumn w_i (y_i - \langle \pE[\ell],X_i\rangle)\langle X_i,u\rangle|w_j = b_j\right]  - o_d(1)
\end{align*}
Using the setting $d > \rho^2$ we have $\rho^2o_d(1) = o_d(1)$ and we simplify the above expression to obtain
\begin{align}\label{eq:3.3.1}
= \Big(\pVar\left[\sumn w_i(y_i - \langle \pE[\ell],X_i\rangle)\langle X_i,u\rangle\right] - \E_{j \sim S_v}\E_{b_j} \pVar\left[\sumn w_i (y_i - \langle \pE[\ell],X_i\rangle)\langle X_i,u\rangle|w_j = b_j\right]\Big)-o_d(1)    
\end{align}
To lower bound the first term above, we apply \pref{thm:onedimrounding} with $Z_i = w_i(y_i - \langle \pE[\ell],X_i\rangle)\langle X_i, u\rangle$.  This immediately gives us,

\begin{align*}
    \pVar\left[\sumn w_i(y_i - \langle \pE[\ell],X_i\rangle)\langle X_i,u\rangle\right] -  \E_{j \sim S_u}\E_{b_j} \pVar\left[\sumn w_i (y_i - \langle \pE[\ell],X_i\rangle)\langle X_i,u\rangle|w_j = b_j\right]
\end{align*}

\begin{align*}
    \geq \frac{\beta\pVar[\sumn w_i (y_i - \langle \pE[\ell],X_i\rangle)\langle X_i,u\rangle]^2}{\sumn\pVar[w_i(y_i - \langle \pE[\ell],X_i\rangle)\langle X_i,u\rangle]}
\end{align*}

\Pnote{slash left bracket and slash right bracket here}

Applying \pref{lem:reg-approx-round} to the numerator we obtain 

\begin{align*}
    \geq  \frac{(1 - o_d(1))\beta\pVar[\iprod{\ell,u}]^2 - o_d(1)}{\sumn\pVar[w_i(y_i - \langle \pE[\ell],X_i\rangle)\langle X_i,u\rangle]}
\end{align*}

Now we upper bound the denominator by

$$ \sumn\pVar[w_i(y_i - \langle \pE[\ell],X_i\rangle)\langle X_i,u\rangle] \leq 2(g\sigma^4B)^\frac{1}{2} + 2B\norm{\cQ}_{nuc} $$

The proof is as follows.  First we use $\pVar(X) \leq \pE[X^2]$ to obtain 

\begin{align*}
    \sumn\pVar(w_i(y_i - \iprod{ \pE[\ell],X_i})\iprod{ X_i,u}) \leq \sumn\pE[w_i(y_i - \langle \pE[\ell],X_i\rangle)^2\langle X_i,u\rangle^2)]
\end{align*}

\Pnote{should we use $\pVar[]$ everywhere instead of $\pVar()$?  It will make it similar to $\pE[]$}

\begin{align*}
    =  \sumn\pE[w_i(y_i - \langle \ell,X_i\rangle + \langle \ell,X_i\rangle - \langle \pE[\ell],X_i\rangle)^2\langle X_i,u\rangle^2)]
\end{align*}
Then we use degree $2$ SoS triangle inequality to obtain 
\begin{equation}\label{rr4}
    \leq  2\pE[\sumn w_i(y_i - \langle \ell,X_i\rangle)^2\langle X_i,u\rangle^2] + 2\pE[\sumn w_i\langle\ell -  \pE[\ell],X_i\rangle^2\langle X_i,u\rangle^2)]
\end{equation}

The first term is upper bounded by pseudoexpectation Cauchy-Schwarz

\Pnote{can we use the following simpler calculation here 

\begin{align*}
    \pE\left[\sumn w_i(y_i - \langle \ell,X_i\rangle)^2\langle X_i,u\rangle^2\right] & =     \pE\left[\sumn w_i^2(y_i - \langle \ell,X_i\rangle)^2\langle X_i,u\rangle^2\right] \\
& =     \left(\pE\left[\sumn w_i^2(y_i - \langle \ell,X_i\rangle)^4\right] \right)^{\frac{1}{2}}\left( \pE\left[\sumn w_i^2\langle X_i,u\rangle^4\right] \right)^{\frac{1}{2}}\\
\end{align*}
}

\begin{align*}
    \pE\left[\sumn w_i(y_i - \langle \ell,X_i\rangle)^2\langle X_i,u\rangle^2\right] & =     \pE\left[\sumn w_i^2(y_i - \langle \ell,X_i\rangle)^2\langle X_i,u\rangle^2\right] \\
& =     \left(\pE\left[\sumn w_i^2(y_i - \langle \ell,X_i\rangle)^4\right] \right)^{\frac{1}{2}}\left( \pE\left[\sumn w_i^2\langle X_i,u\rangle^4\right] \right)^{\frac{1}{2}}\\
\end{align*}

Then by degree $2$ SoS Cauchy-Schwarz, followed by applying the fourth moment constraints on noise (4) we obtain  

\Pnote{ along with referring to $4th$ moment constraint on noise, it might be clearer if we number the constraints of the SDP, and refer to a particular constraint of the SDP that is used, directly with numbering}

$$\leq (g\sigma^4)^{1/2}\pE[\iprod{\sumn w_iX_i^{\otimes 2}(X_i^{\otimes 2})^T, u^{\otimes 2}(u^{\otimes 2})^T}]^\frac{1}{2} 
$$

$$
= (g\sigma^4)^\frac{1}{2}\pE[\iprod{\sumn w_iX_i^{\otimes 2}(X_i^{\otimes 2})^T - M_4, u^{\otimes 2}(u^{\otimes 2})^T} + \iprod{M_4,u^{\otimes 2}(u^{\otimes 2})^T}]^\frac{1}{2} 
$$
Then applying Cauchy-Schwarz, followed by applying the fourth moment constraints on the covariates (3) we obtain 

$$
\leq (g\sigma^4)^\frac{1}{2}\pE[\norm{\sumn w_iX_i^{\otimes 2}(X_i^{\otimes 2})^T - M_4}_F^2 + B]^\frac{1}{2} 
\leq (2g\sigma^4B)^\frac{1}{2}
$$

Next we upper bound the second term in \ref{rr4} by SoS Cauchy Schwarz

\begin{align*}
    \pE[\sumn w_i\langle\ell -  \pE[\ell],X_i\rangle^2\langle X_i,u\rangle^2)]  = \pE[\langle (\ell - \pE[\ell])^{\otimes 2}(u^{\otimes 2})^T, \sumn w_iX_i^{\otimes 2}(X_i^{\otimes 2})^T \rangle]
\end{align*}

\begin{align*}
     = \pE[\langle (\ell - \pE[\ell])^{\otimes 2}(u^{\otimes 2})^T, \sumn w_iX_i^{\otimes 2}(X_i^{\otimes 2})^T  - M_4\rangle] + \pE[\langle (\ell - \pE[\ell])^{\otimes 2}(u^{\otimes 2})^T, M_4\rangle]
\end{align*}
Applying SoS Cauchy-Schwarz to the first term we obtain
\begin{align*}
     = \pE[\norm{\ell - \pE[\ell]}^2\norm{\sumn w_iX_i^{\otimes 2}(X_i^{\otimes 2})^T - M_4}_F^2] + \pE[\langle (\ell - \pE[\ell])^{\otimes 2}(u^{\otimes 2})^T, M_4\rangle]
\end{align*}
Then applying the fourth moment constraints on the covariates (3) and applying the definition of $\norm{\cQ}_{nuc}$ we obtain 

\begin{align*}
     = \epsilon\norm{\cQ}_{nuc} + \pE[\langle (\ell - \pE[\ell])^{\otimes 2}(u^{\otimes 2})^T, M_4\rangle]
\end{align*}
We upper bound the second term above using the assumption upper bounding the fourth injective norm of the covariates.   
\begin{align*}
    \pE[\langle (\ell - \pE[\ell])^{\otimes 2}(u^{\otimes 2})^T, M_4\rangle] \leq B\norm{\cQ}_{nuc}
\end{align*}
For $\cN(0,I)$, we have $B = 3$.  Plugging both terms back into \ref{rr4}, we obtain

  \begin{align*}
    \pVar_\zeta[\iprod{\ell,v}] -  \E_{j \sim S_v}\E_{b_j} \pVar_\zeta[\iprod{\ell,v}|w_j = b_j]
\end{align*}

\begin{align*}
    \geq  \frac{(1-o_d(1))\beta\pVar_\zeta[\iprod{\ell,v}]^2 - o_d(1)}{2\sigma^2(gB)^\frac{1}{2} + 2B\norm{\cQ}_{nuc}} - o_d(1)
    =  \frac{\beta\pVar_\zeta[\iprod{\ell,v}]^2 }{2\sigma^2(gB)^\frac{1}{2} + 2B\norm{\cQ}_{nuc}} - o_d(1)
\end{align*}

\Pnote{in this bound, can we just use $B (1+\epsilon_4)$ by $2\norm{M_4}_inj$}
Using the assumption $\norm{\cQ}_{nuc} > \sigma^2\sqrt{g}$ and setting $B = 3$ we have 

\begin{align*}
    \pVar_\zeta[\iprod{\ell,v}] -  \E_{j \sim S_v}\E_{b_j} \pVar_\zeta[\iprod{\ell,v}|w_j = b_j]
    \geq  \Omega\Big(\frac{\beta\pVar_\zeta[\iprod{\ell,v}]^2}{\norm{\cQ}_{nuc}}\Big)
\end{align*}

\end{proof}

\Pnote{We should replace all instances of $||$ by $\norm{}$}

\subsection{Snapping for Regression} \label{sec:regressionsnapping}

\restatelemma{lem:regressionSnapping}

\begin{proof}

Let $u \in S^{d-1}$, we have by linearity 
\begin{align*}
    \pE\left[\sumn w_iw_i'\right]\langle \pE[\ell] - \ell',u\rangle^2 = \pE\left[\sumn w_iw_i'\langle \pE[\ell] - \ell',u\rangle^2\right]
\end{align*}
And by degree $2$ SoS triangle inequality 
\begin{align*}
    =\pE[\sumn w_iw_i'\langle \pE[\ell] - \ell + \ell - \ell',u\rangle^2] \leq 2
    \pE[\sumn w_iw_i'\langle \pE[\ell] - \ell,u\rangle^2]  + 2\pE[\sumn w_iw_i'\langle\ell - \ell',u\rangle^2]
\end{align*}
The following expression is a sum of squares  $\cR \SoSp{2} \Big\{\sumn (1 - w_i)w_i'\iprod{\pE[\ell] - \ell,u}^2 \geq 0\Big\}$ so we add it to the right hand side to obtain 

\begin{align*}
    \leq 2
    \pE[\langle \pE[\ell] - \ell,u\rangle^2]  + 2\pE[\sumn w_iw_i'\langle\ell - \ell',u\rangle^2]
\end{align*}
applying degree $2$ SoS Cauchy-Schwarz to the second term we obtain, 

\begin{align*}
    \leq 2
    \pE[\langle \pE[\ell] - \ell,u\rangle^2]  + 2\pE[\sumn w_iw_i'||\ell - \ell'||^2]
\end{align*}
Consider the second term above. By the properties of $(c,D(\eta))$-SoS-anticoncentration (see \pref{def:anticoncentration}) we upper bound by, \Pnote{expand on this citation, mention definition and lemma}

\begin{align*}
    &\leq 2
    \pE[\langle \pE[\ell] - \ell,u\rangle^2]  + 2\left(c\rho\eta + \frac{1}{\eta^2}\pE\left[\sumn w_iw_i'\langle
    \ell - \ell',X_i\rangle^2\right]\right) \\
    &\leq 2
    \pE[\langle \pE[\ell] - \ell,u\rangle^2]  + 2\left(\rho^2\eta + \eta^2\pE\left[\sumn w_iw_i'(\langle
    \ell,X_i\rangle - y_i)^2 + (\langle
    \ell',X_i\rangle - y_i)^2\right]\right)
\end{align*}

\Pnote{ what is $\delta$ in the above calculation?}
By SoS triangle inequality

\begin{align*}
    \leq 2
    \pE[\langle \pE[\ell] - \ell,u\rangle^2]  + 2(\rho^2\eta + \eta^22(\pE[\sumn w_iw_i'(\langle
    \ell,X_i\rangle - y_i)^2] + 2\pE[\sumn w_iw_i'(\langle
    \ell',X_i\rangle - y_i)^2]))
\end{align*}
Using the fact that $\cP \SoSp{2} \Big\{\sumn(1 - w_i)w_i'(\langle
    \ell,X_i\rangle - y_i)^2 \geq 0, \sumn(1 - w_i')w_i(\langle
    \ell,X_i\rangle - y_i)^2 \geq 0\Big\}$
we add in both polynomials to obtain
$$ \leq 2
    \pE[\langle \pE[\ell] - \ell,u\rangle^2]  + 2(\rho^2\eta + \eta^2 2(\pE[\sumn w_i(\langle
    \ell,X_i\rangle - y_i)^2] + 2\pE[\sumn w_i'(\langle
    \ell',X_i\rangle - y_i)^2]))
$$
Applying the SDP noise constraint (4) we obtain 
\begin{align*}
    \leq 2
    \pE[\langle \pE[\ell] - \ell,u\rangle^2]  + (2c\rho^2\eta + \frac{8\sigma^2}{\eta^2})
\end{align*}
Thus far we've shown in degree $D(\eta)$ the following inequality
$$
 \pE[\sumn w_iw_i']\langle \pE[\ell] - \ell',u\rangle^2 \leq 2
    \pE[\langle \pE[\ell] - \ell,u\rangle^2]  + (2c\rho^2\eta + \frac{8\sigma^2}{\eta^2})
$$
This inequality holds for all $u \in S^{d-1}$, in particular for the unit vector $u$ along $\pE[\ell] - \ell'$ we have 
$$
 \pE[\sumn w_iw_i']\norm{\pE[\ell] - \ell'}^2 \leq 2
    \max_{u \in S^{d-1}}\pE[\langle \pE[\ell] - \ell,u\rangle^2]  + (2c\rho^2\eta + \frac{8\sigma^2}{\eta^2})
$$
Dividing both sides by  $\pE[\sumn w_iw_i']$ and taking a square root we obtain 

$$
 \norm{\pE[\ell] - \ell'} \leq \sqrt{ \frac{2
    \max_{u \in S^{d-1}}\pE[\langle \pE[\ell] - \ell,u\rangle^2]  + (2c\rho^2\eta + \frac{8\sigma^2}{\eta^2})}{\pE[\sumn w_iw_i']}}
$$
Plugging in the assumptions on frobenius minimization \ref{fm} and variance reduction \ref{vr} we obtain 
\Pnote{refer to the equations here}
$$
 \norm{\pE[\ell] - \ell'} \leq \sqrt{ \frac{4c\rho^2\eta + \frac{8\sigma^2}{\eta^2}}{\beta}}
$$
Since $\eta$ is any constant in $[0,\frac{1}{2}]$ we conclude by writing 
$$
 \norm{\pE[\ell] - \ell'} \leq \sqrt{ \frac{\rho^2\eta + O(\frac{\sigma^2}{\eta^2})}{\beta}}
$$

\end{proof}

\section{Certifying Anticoncentration}\label{sec:anticoncentration}

Anticoncentration is a measure of the "spread" of a distribution.  For any distribution $\cD$, let $\eta$ be a parameter $0 < \eta < \frac{1}{2}$.  If the probability mass of $\cD$ contained in the $\eta$ interval around the origin is small, than $\cD$ is anticoncentrated.  For example, in the case of $\cD = \cN(0,1)$, the mass of $\cD$ in any $\eta$ interval about the origin is upper bounded by $\frac{2}{\sqrt{2\pi}}\eta$.  This property of the probability mass decaying linearly with $\eta$ as $\eta$ goes to zero is what motivates the following definition.

\begin{definition}\label{def:anticoncentration}
    A probability distribution $\cD$ over $\R^d$ is said to be $c$-anticoncentrated if for any $0 < \eta < \frac{1}{2}$ there exists $\tau \leq c\eta$ such that for any measurable subset $\cE \in \R^n$, and for all $v \in \R^d$ with $\norm{v} \leq 1$, we have that
    $$ \E[ \iprod{X,v}^2 \cdot \Ind[\cE] ] \geq \eta^2 \cdot \Pr[\cE]  \cdot \norm{v}^2 - \eta^2\tau$$ 
\end{definition}
\Pnote{$\Omega$ is undefined here. We should probably define $\cE$ to be a measurable subset of $\R^n$}

We now state the SoS version of anticoncentration
\begin{definition}\label{def:anticoncentration}
Let $D: [0,1/2] \to \N$.
A probability distribution $\cD$ over $\R^d$ is said to $(c, D(\eta))$-SoS-anticoncentrated, 
If for any $0 < \eta < \frac{1}{2}$ there exists $\tau \leq c\eta$ and there exists a constant $k \in \N$ such that for all $N > d^k$,
with probability $1-d^{-k}$, over samples $x_1,\ldots, x_N \sim \cD$ the following polynomial system

$$
\cP = \left\{
        \begin{array}{ll}
            w_i^2 = w_i & i \in [N]\\
    \norm{v}^2 \leq \rho^2 \\
    \norm{\frac{1}{N}\sum_{i=1}^N X_i^{\otimes \frac{t}{2}}(X_i^{\otimes \frac{t}{2}})^T - M_t} < \epsilon& t \in [k]\\
        \end{array}
    \right.
$$
yields a degree $D(\eta)$ SoS proof of the following inequality 

\begin{align*}
    \cP\SoSp{D(\eta)}\Big\{\frac{1}{N}\sum_{i = 1}^N w_i \iprod{X_i,v}^2 \geq \eta^2 (\frac{1}{N}\sum_i w_i) \norm{v}^2 - \eta^2 \tau \rho^2\Big\}
\end{align*}
\end{definition}

\Pnote{good to remind the reader what $M_t$ denotes in the above definition.  Perhaps, we should uniformly use $M_4(\cD)$?}

\begin{theorem}\label{thm:anti-sufficient} (Sufficient conditions for SoS anti-concentration)
If the degree $D(\eta)$ empirical moments of $\cD$ converge to the corresponding true moments $M_t$ of $\cD$, that is for all $t \leq D(\eta)$ 
$$
\lim_{N \rightarrow \infty}\Norm{\frac{1}{N}\sum_{i=1}^N X_i^{\otimes \frac{t}{2}}(X_i^{\otimes \frac{t}{2}})^T - M_t} = 0 
$$
And if there exists a uni-variate polynomial $I_{\eta}(z) \in \R[z]$ of degree at most $D(\eta)$ such that 
\begin{enumerate}
    \item $I_{\eta}(z) \geq 1-\frac{z^2}{\eta^2\rho^2}$ for all $z \in \R$.
    \item $\cP\SoSp{D(\eta)} \Big\{\norm{v}^2 \cdot \E_{x \in \cD}[I_{\eta}(\iprod{v, x})] \leq c\eta\rho^2\Big\}$.
\end{enumerate}
Then $\cD$ is $(c,D(\eta))$ certifiably anticoncentrated.

\end{theorem}

\begin{lemma}\torestate{\label{lem:normal-anti}
For every $d \in \N$, the standard Gaussian distribution $\cN(0,I_d)$ is $(c,O(\frac{1}{\eta^4}))$-SoS-anticoncentrated. In particular there exists a construction for $c \leq 2\sqrt{e}$}    
\end{lemma}

First we will prove \pref{thm:anti-sufficient}

\begin{proof} (\pref{thm:anti-sufficient})
First, it is a standard fact that every uni-variate polynomial inequality has a sum of squares proof.  More precisely, for any $p(x) \in \R[x]$ satisfying $p(x) \geq 0$, then it is true that $p(x) \SoSge_{\text{deg}(p(x))} 0$.  Furthermore, this is also true over any interval $[a,b] $

\begin{fact}\label{fact:univar}
Let $a < b$. Then, a degree $2d$ polynomial $p(x)$ is non-negative on $[a, b]$, if and only if it can be written as
$$
\left\{
\begin{array}{ll}
     & p(x) = s(x) + (x - a)(b - x)t(x), \text{    if deg(p) is even}\\
    & p(x) = (x - a)s(x) + (b - x)t(x), \text{     if deg(p) is odd}\\
\end{array}
\right.
$$
where $s(x)$, $t(x)$ are SoS. In the first case, we have $deg(p) = 2d$, and $deg(s) \leq 2d$, $deg(t) \leq 2d - 2$. In
the second, $deg(p) = 2d + 1$, and $deg(s) \leq 2d$, $deg(t) \leq 2d$.
\end{fact}
In light of this fact, we use \pref{thm:anti-sufficient} condition 1 to lower bound $\iprod{X_i,v}^2$ by 

\begin{align*}
\cP \SoSp{D(\eta)} \langle X,v\rangle^2 \geq \eta^2\rho^2(1 - I_\eta(\langle X_i,v\rangle))
\end{align*}
Therefore, 
\begin{align*}
    \cP \SoSp{D(\eta)}\sumn w_iw_i'\langle X,v\rangle^2 \geq \sumn w_iw_i'\eta^2\rho^2 (1 - I_{\delta}(\langle X_i,v\rangle))
\end{align*}
Then using the certificate that $\{\norm{v}^2 < \rho^2\}$ we obtain
\begin{align*}
    \cP \SoSp{D(\eta)}\sumn w_iw_i'\eta^2\rho^2(1 - I_{\eta}(\langle X_i,v\rangle)) \geq \sumn w_iw_i'\eta^2\norm{v}^2 (1 - I_{\eta}(\langle X_i,v\rangle))
\end{align*}

\begin{align*}
    = \sumn w_iw_i'\eta^2\norm{v}^2 - \sumn w_iw_i'\eta^2\norm{v}^2I_{\eta}(\langle X_i,v\rangle)
\end{align*}
Then using the fact that $I_\eta(\langle X_i,v\rangle)$ is SoS and $\{w_i^2 = w_i\} \SoSp{} (1 - w_i) \SoSge_2 0$, we subtract $\sumn w_i'(1-w_i)I_\eta(\langle X_i,v\rangle)$ to obtain 

\begin{align*}
    \geq \sumn w_iw_i'\eta^2\norm{v}^2 - \eta^2 \norm{v}^2\sumn w_i'I_{\eta}(\langle X_i,v\rangle)
\end{align*}

Expanding out $I_\delta(\langle X_i,v\rangle)$ as a degree $D(\eta)$ polynomial with coefficients $\alpha_1,...,\alpha_{D(\eta)}$ we have

\begin{align*}
    I_\eta(\langle X_i,v\rangle) = \sum_{t=1}^T \alpha_t \langle X_i,v\rangle^t
\end{align*}
We want replace the empirical average $\sumn w_i' I_\eta(\langle X_i,v\rangle)$ with the expectation $\E_{X \sim N(0,I)} I_\eta(\langle X,v \rangle)$ and bound the error term.  Indeed, we know that
\begin{align*}
     \sumn w_i'I_\eta(\langle X_i,v\rangle) = \sum\limits_{t=1}^T \alpha_t\sumn w_i'\langle X_i,v\rangle^t  = \sum\limits_{t=1}^{D(\eta)} \alpha_t\big \langle \sumn w_i' X_i^{\otimes \frac{t}{2}}(X_i^{\otimes\frac{t}{2}})^T, v^{\otimes \frac{t}{2}}(v^{\otimes\frac{t}{2}})^T \big\rangle
\end{align*}

\begin{align*}
    = \sum\limits_{t=1}^T \alpha_t\big \langle \sumn w_i' X_i^{\otimes \frac{t}{2}}(X_i^{\otimes\frac{t}{2}})^T - \E_{X \sim N(0,I)} X^{\otimes \frac{t}{2}}(X^{\otimes\frac{t}{2}})^T, v^{\otimes \frac{t}{2}}(v^{\otimes\frac{t}{2}})^T \big\rangle + \sum\limits_{t=1}^{D(\eta)} \alpha_t\big \langle \E_{X \sim N(0,I)} X^{\otimes \frac{t}{2}}(X^{\otimes\frac{t}{2}})^T, v^{\otimes \frac{t}{2}}(v^{\otimes\frac{t}{2}})^T \big\rangle
\end{align*}

Then by degree $D(\eta)$ SoS Cauchy Schwarz we obtain
\begin{align*}
    \SoSle  \sum\limits_{t=1}^{D(\eta)} \alpha_t\norm{\sumn w_i' X_i^{\otimes \frac{t}{2}}(X_i^{\otimes\frac{t}{2}})^T - M_t}_F^2 \norm{v}^t + \sum\limits_{t=1}^{D(\eta)}\alpha_t\E_{X \sim N(0,I)}\langle X,v\rangle^t
\end{align*}
Thus for our setting of $N$ and $d$ we obtain,
\begin{align}\label{eq:anti1}
 = \E_{X \sim N(0,I)}I_\eta(\iprod{X_i,v}) + o_d(1)
\end{align}
Note that it is important that the coefficients of $I_\eta(z)$ are chosen independently of $d$ or at the very least don't grow too fast with respect to $d$. Our final bound is, 

\begin{align*}
    \cP \SoSp{D(\eta)}\sumn w_iw_i'\langle X,v\rangle^2 \geq \sumn w_iw_i'\eta^2\norm{v}^2 - \eta^2 \norm{v}^2\E_{z \sim N(0,\norm{v}^2)} I_{\eta}(z) + o_d(1)
\end{align*}
Applying sufficient condition 2 we obtain 
\begin{align*}
    \cP\SoSp{D(\eta)}\Big\{\frac{1}{N}\sum_{i = 1}^N w_i \iprod{X_i,v}^2 \geq \eta^2 (\frac{1}{N}\sum_i w_i) \norm{v}^2 - \eta^2 \tau\rho^2 \Big\}
\end{align*}
as desired.  
\end{proof}

\section{Certifiably Anticoncentrated Distributions} \label{sec:anticoncentration-distributions}
\restatelemma{lem:normal-anti}
\begin{proof}
By \pref{thm:anti-sufficient} it suffices to exhibit a polynomial $I_\eta(x)$ satisfying

\begin{enumerate}
    \item $I_\eta(x) \geq 1 - \frac{x^2}{\eta^2\rho^2}$
    \item $\cP \SoSp{O(\frac{1}{\eta^4})}\Big\{\norm{v}^2\E_{x \sim N(0,I)} I_\eta(\iprod{X,v}) \leq  c\eta\rho^2 \Big\}$
\end{enumerate}
\Pnote{what is the extra $\nu$ on the right hand side?}

Firstly,  without loss of generality the scaling $\rho$ can be set to $1$ so that $\rho = 1$ and $\norm{v} \leq 1$.  This is because any polynomial $I_\eta(x)$ satisfying conditions 1 and 2 for $\rho = 1$ and  $\norm{v} \leq 1$ can be reparameterized as $I(x') = I_\eta(\frac{x'}{\rho})$ and satisfy conditions 1 and 2 for $\norm{v} \leq \rho$ for general $\rho$.       

Next we observe that owing to the spherical symmetry of the standard Gaussian we have $\norm{v}^2\E_{x \sim N(0,I)} I_\eta(\iprod{X,v})$ is a spherically symmetric polynomial in $X$ which implies it is a polynomial in $\norm{v}$.  Thus define

$$H(\norm{v})  := \norm{v}^2\E_{x \sim N(0,I)} I_\eta(\iprod{X,v}) = \norm{v}^2\E_{x \sim N(0,\norm{v}^2)} I_\eta(x)$$

Furthermore we have $\{\norm{v}^2 \leq 1\} \in \cP$ and $\norm{v} \geq 0$ is SoS.  Therefore, it suffices to prove the inequality $H(\norm{v}) \leq \eta$ and \pref{fact:univar} implies condition 2.  
\Pnote{we might have to add "Fact" to prettyref macro?}
Now we construct $I_\eta(x)$, which we refer to as the anticoncentration polynomial.  Note that the indicator function of the $[-\eta,\eta]$ interval satisfies both anticoncentration conditions.  The idea is to approximate the indicator function with a polynomial.    
\Pnote{"anticoncentration polynomial" terminology unknown to reader}
It is difficult to directly approximate the indicator function as it is not continuous.  Thus we dominate the indicator by a scaled Gaussian denoted $f(x)$ which satisfies the anticoncentration conditions.  
\Pnote{reader doesn't know why we are approximating indicator functions, and which indicator function we are dealing with}
We then interpolate an explicit sum of squares polynomial through $f(x)$ denoted $I_\eta(x)$.  The key here is that any uni variate positive polynomial blows up at its tails.  Thus, we must prove the approximation error of $|f(x) - I_\eta(x)|$ is small for some interval around the origin, and far away from the origin that the decay of the Gaussian tail dominates the growth of the approximation error.  

We note that there are many different strategies to construct polynomials  satisfying the above criterion, and we will satisfy ourselves with proving the Gaussian is $(2\sqrt{e},O(\frac{1}{\eta^4}))$-certifiably anticoncentrated.

First let $f(x) = \sqrt{e}\exp(-\frac{x^2}{2\eta^2})$.  For simplicity we will design $f(x)$ such that  $f(\pm\eta) = 1$ to satisfy the first anticoncentration condition.  Checking the second condition we find that 

\begin{align*}
    \norm{v}^2\E_{x \sim N(0,\norm{v}^2)} f(x) = \norm{v}^2\int \frac{\sqrt{e}}{\sqrt{2\pi}\norm{v}} \exp(-\frac{x^2}{2\eta^2} - \frac{x^2}{2\norm{v}^2}) dx
\end{align*}
\Pnote{brackets need cleaning}

\begin{align*}
     = \norm{v}^2\frac{\sqrt{e}}{\sqrt{2\pi}\norm{v}}\int \exp\Big(-\frac{x^2}{2(\frac{\eta^2\norm{v}^2}{\eta + \norm{v}^2})} \Big) dx = \norm{v}^2\frac{\eta \sqrt{e}}{\sqrt{\eta^2 + \norm{v}^2}}
\end{align*}

\begin{align*}
    \leq \norm{v}^2 \frac{\eta \sqrt{e}}{\norm{v}} \leq \eta\sqrt{e}
\end{align*}

Where in the last inequality we used $0 \leq \norm{v}^2 \leq 1$.

Intuitively, if we interpolate a sum of squares polynomial $I_\eta(x)$ that closely approximates $f(x)$ in an interval around the origin, then $\E_{x \sim N(0,\norm{v}^2)} I_\eta(x) \approx \E_{x\sim N(0,\norm{v}^2)} f(x)$.  Let $(x_0,x_1,...,x_n)$ be evenly spaced points at intervals of length $\nu\eta^r$ ranging from $[-\frac{\nu n\eta^r}{2},\frac{\nu n\eta^r}{2}]$ where we eventually set $\nu$ to be a constant and $r = 4$.  
\Pnote{what is $r$ here?  we need to specify it or remind the reader}
Let $(y_0,...,y_n)$ be the set of evaluations $y_i = f(x_i)$.  Let $I_\eta(x)$ be the following degree $2n$ polynomial.

\begin{align*}
    I_\eta(x) = \frac{(x-x_1)^2(x - x_2)^2...(x-x_n)^2}{(x_0 - x_1)^2(x_0 - x_2)^2...(x_0 - x_n)^2}y_0 + \frac{(x-x_0)^2(x - x_1)^2...(x-x_n)^2}{(x_1 - x_0)^2(x_1 - x_2)^2...(x_1 - x_n)^2}y_1 + ... \\+ \frac{(x-x_0)^2(x - x_1)^2...(x-x_{n-1})^2}{(x_n - x_0)^2(x_n - x_1)^2...(x_n - x_{n-1})^2}y_n
\end{align*}

\begin{align*}
    = \sum\limits_{i=1}^n\Big(\prod_{\substack{{0 \leq j \leq n}\\ {i \neq j}}} \frac{(x - x_i)^2}{(x_i - x_j)^2}\Big)y_j
\end{align*}

$I_\eta(x)$ is the standard interpolation polynomial where each term is squared so as to be a sum of squares.  Let $R_{2n}(y)$ be the error term over the interval $[-y,y]$ be $R_{2n}(y) = \max_{x \in [-y,y]}|f(x) - I_\eta(x)|$.  It is easy to show the interpolation error is

\begin{align*}
    R_{2n}(y) = \frac{1}{(2n + 1)!}\max\limits_{x \in [-y,y]}|f^{2n + 1}(x)|\prod_{i=0}^n (x - x_i)^2
\end{align*}

One way to prove the above equality is to think of the construction of $I_\eta(x)$ as follows.   Let $\Tilde{I}(x)$ be the unique degree $2n$ interpolation of points $\{(x_i,f(x_i))\}_{i \in [n]}$ and $\{(x_i', f(x_i'))\}_{i \in [n]}$ , which is not necessarily a sum of squares.  
\Pnote{interpolation of points is a little vague, interpolation through points $\{(x_i,f(x_i))\}_{i \in [n]}$ and $\{(x_i', f(x_i'))\}_{i \in [n]}$}.
It is a standard fact in polynomial approximation theory, \cite{Sauer1997}, that the error $\Tilde{I}_{2n}(y) = \max_{x \in [-y,y]}|f(x) - \Tilde{I}(x)|$ has the form.  
\Pnote{some reference for the fact will be useful here}
\begin{align*}
    \Tilde{R}_{2n}(y) = \frac{1}{(2n+1)!}\max\limits_{x \in [-y,y]}|f^{2n+1}(x)|\prod_{i=0}^n (x - x_i)\prod_{i=0}^n (x - x_i')
\end{align*}

It is easy to check that $I_\eta(x) = \lim_{(x_0',...,x_n') \rightarrow (x_0,...,x_n)}\Tilde{I}(x)$.  Thus 

\begin{align*}
R_{2n}(y) = \lim_{(x_0',...,x_n') \rightarrow (x_0,...,x_n)}\Tilde{R}(y) = \frac{1}{(2n+1)!}\max\limits_{x \in [-y,y]}|f^{2n+1}(x)|\prod_{i=0}^n (x - x_i)^2
\end{align*}

as desired.  

Now we verify anticoncentration condition 2

\begin{align}
    \norm{v}^2\E_{x \sim N(0,\norm{v}^2)} I_\eta(x) & \leq \norm{v}^2\E_{x \sim N(0,\norm{v}^2)} |f(x) - I_\eta(x)| + \norm{v}^2\E_{x \sim N(0,\norm{v}^2)}f(x) \nonumber \\
    & \leq  \norm{v}^2\E_{x \sim N(0,\norm{v}^2)} |f(x) - I_{\eta}(x)| + \sqrt{e}\eta
\end{align}

\Pnote{I don't follow all the details below.  Why can we bound $R_{2n}(y)$ in second term by $R_{2n}( \nu \eta^r n/2)$?}

\begin{align*}
     =  \norm{v}^2 \big(\int_{-\infty}^{-\frac{\nu\eta^r n}{2}} R_{2n}(y) \frac{1}{\sqrt{2\pi}\norm{v}}\exp(-\frac{y^2}{2\norm{v}^2})dy +  \int_{-\frac{\nu\eta^r n}{2}}^{\frac{\nu\eta^r n}{2}}R_{2n}(y)\frac{1}{\sqrt{2\pi}\norm{v}}\exp(-\frac{y^2}{2\norm{v}^2})dy \\ +  \int_{\frac{\nu\eta^r n}{2}}^{\infty} R_{2n}(y) \frac{1}{\sqrt{2\pi}\norm{v}}\exp(-\frac{y^2}{2\norm{v}^2})dy\big) + \sqrt{e}\eta
\end{align*}

Since we defined $R_{2n}(y)$ to be the maximum error in the $[-y,y]$ interval, it is monotonic, and we upper bound it by its evaluation at its rightmost endpoint $R_{2n}(\frac{\nu\eta^rn}{2})$.   

\begin{align*}
     \leq  \norm{v}^2 \big(\int_{-\infty}^{-\frac{\nu\eta^r n}{2}} R_{2n}(y) \frac{1}{\sqrt{2\pi}\norm{v}}\exp(-\frac{y^2}{2\norm{v}^2})dy +  R_{2n}(\frac{\nu\eta^r n}{2})\int_{-\frac{\nu\eta^r n}{2}}^{\frac{\nu\eta^r n}{2}}\frac{1}{\sqrt{2\pi}\norm{v}}\exp(-\frac{y^2}{2\norm{v}^2})dy \\ +  \int_{\frac{\nu\eta^r n}{2}}^{\infty} R_{2n}(y) \frac{1}{\sqrt{2\pi}\norm{v}}\exp(-\frac{y^2}{2\norm{v}^2})dy\big) + \sqrt{e}\eta
\end{align*}

\begin{align*}
     =  \norm{v}^2 \big( R_{2n}(\frac{\nu\eta^r n}{2}) +  2\int_{\frac{\nu\eta^r n}{2}}^{\infty} R_{2n}(y) \exp(-\frac{y^2}{2})dy\big) + \sqrt{e}\eta
\end{align*}

Thus it suffices to show 

\begin{align*}
R_{2n}(\frac{\nu\eta^r n}{2}) +  2\int_{\frac{\nu\eta^r n}{2}}^{\infty} R_{2n}(y) \exp(-\frac{y^2}{2})dy \leq \sqrt{e}\eta    
\end{align*}

Without loss of generality let $\rho = 1$.  Let's start with $R_{2n}(\frac{\nu\eta^r n}{2})$
\Pnote{its not obvious why we can set $\rho = 1$ without loss of generality here}

\begin{align}\label{eq:anti1}
    R_{2n}(\frac{\nu\eta^r n}{2}) \leq \frac{1}{(2n+1)!} \max_{x \in [-y,y]}|f^{2n+1}(x)| \prod_{i=0}^n (x - x_i)^2
\end{align}

Since $f(x)$ is a scaled Gaussian, we have directly from its Taylor expansion $$\max_{x \in \R}|f^{2n+1}(x)| < \max_{x \in \R}|f^{2n+2}(x)| = |f^{2n+2}(0)| = \frac{(2n+2)!}{(n+1)!}(\frac{1}{2\eta^2})^{2n+2}$$.

Plugging the above bound into \ref{eq:anti1} we obtain 

\begin{align*}
     R_{2n}(\frac{\nu\eta^r n}{2}) < \frac{\nu^2\eta^{2r}(2^2\nu^2\eta^{2r})...(n^2\nu^2\eta^{2r})}{(2n+1)!} |f^{2n+2}(0)| = \frac{\nu^{2n}\eta^{2rn}(n!)^2}{(2n+1)!} |f^{2n+2}(0)| =  \frac{\nu^{2n}\eta^{2rn} (n!)^2}{(2n+1)!} \frac{(2n+2)!}{(n+1)!}(\frac{1}{2\eta^2})^{2n+2}
\end{align*}  

$$=  2\nu^{2n}\eta^{2rn} (n!)\big(\frac{1}{2\eta^2}\big)^{2n+2} = 2(\frac{1}{4\eta^4})\big(\frac{\nu\eta^{r-2}}{2}\big)^{2n}n! = 2(\frac{1}{4\eta^4})\sqrt{2 \pi n}\big(\frac{n}{e}\big)^n\big(\frac{\nu\eta^{r-2}}{2}\big)^{2n} = 2(\frac{1}{4\eta^4})\frac{\sqrt{2\pi n}}{2^{2n}e^n} (\nu\eta^{r-2}\sqrt{n})^{2n}$$
Where the factorial approximation is Stirling's.  
Thus a sufficient condition for error decay is $\nu\eta^{r-2}\sqrt{n} \leq 1$.  Then for the benefit of tail error decay, we will set $n$ to saturate the center interval error $n := \frac{1}{\sqrt{e}\nu^2\eta^{2(r-2)}}$ where the $\sqrt{e}$ will be to accommodate for some discrepancy in error in the tail bound. Intuitively, the larger the value of $n$ the further our the interpolation points, and the better the Gaussian tail dominates the polynomial growth in error.

Next we show the tail error is small.  
 
\begin{align*}
\int_{\frac{\nu\eta^r n}{2}}^{\infty} R_{2n}(y) \exp(-\frac{y^2}{2})dy \leq 
\int_{\frac{\nu\eta^r n}{2}}^{\infty} \frac{(y + \frac{\nu\eta^r n}{2})^{2n}}{(2n+1)!}|f^{2n+1}(0)| \exp(-\frac{y^2}{2})dy   
\end{align*}

\begin{align*}
\leq \int_{\frac{\nu\eta^r n}{2}}^{\infty} \frac{(y + \frac{\nu\eta^r
n}{2})^{2n}}{(2n+1)!}\frac{(2n+2)!}{(n+1)!}\big(\frac{1}{2\eta^2}\big)^{2n+2} \exp(-\frac{y^2}{2})dy   
\end{align*}

\begin{align*}
\leq 2\int_{\frac{\nu\eta^r n}{2}}^{\infty} \frac{(y + \frac{\nu\eta^r n}{2})^{2n}}{n!}\big(\frac{1}{2\eta^2}\big)^{2n+2} \exp(-\frac{y^2}{2})dy   
\end{align*}

The integrand evaluated at $y = \frac{\nu\eta^r n}{2}$ is  

$$\leq \frac{1}{\sqrt{2\pi n}}\big(\frac{1}{4\eta^4}\big)\big(\frac{\nu\eta^{r-2}\sqrt{en}}{2}\big)^{2n}\exp(-\frac{(\nu\eta^rn)^2}{8}) $$

By our choice of $n$ we have both the exponential and the error term falling to zero rapidly. For $r = 4$ and $\nu = 1/100$ we have $n = O(\frac{1}{\eta^4})$ for a degree $2n = O(\frac{1}{\eta^4})$ polynomial. 
\end{proof}

The following lemma establishes that the sufficient conditions of \pref{thm:anti-sufficient} are naturally extended under linear transformations of the data set.  
\begin{lemma} (Anticoncentration under Linear Transformation)
Let $\cD$ be a $(c,D(\eta))$ certifiably anticoncentrated distribution.  Let $I_\eta(z) \in \R[z]$ be a uni-variate polynomial satifying the conditions of \pref{thm:anti-sufficient}.  Let $x \sim \cD$ be a random variable drawn from $\cD$. Then for any invertible linear transformation $A \in \R^{d \times d}$, we denote the distribution of  $Ax$ as $A(\cD)$.  Let $\Sigma = AA^T$ be the covariance of $A(\cD)$ with eigenvalues $\lambda_1,\lambda_2,...,\lambda_d$.  Then $A(\cD)$ is $(c\frac{\lambda_1^{3/2}}{\lambda_d^{3/2}},D(\eta))$ certifiably anticoncentrated.     
\end{lemma}

\begin{proof}
In light of this fact, we use \pref{thm:anti-sufficient} condition 1 to lower bound $\iprod{X_i,v}^2$ by 

\begin{align*}
\cP \SoSp{D(\eta)} \langle X,v\rangle^2 \geq \eta^2(1 - I_\eta(\langle X_i,v\rangle))
\end{align*}

Therefore, 
\begin{align*}
    \cP \SoSp{D(\eta)}\frac{1}{M} \sum\limits_{i=1}^N w_iw_i'\langle X,v\rangle^2 \geq \frac{1}{M} \sum\limits_{i=1}^N w_iw_i'\eta^2 (1 - I_{\delta}(\langle X_i,v\rangle))
\end{align*}

Then using the certificate that $\{\norm{v}^2 < 1\}$ we obtain

\begin{align*}
    \cP \SoSp{D(\eta)}\frac{1}{M} \sum\limits_{i=1}^N w_iw_i'\eta^2 (1 - I_{\eta}(\langle X_i,v\rangle)) \geq \frac{1}{M} \sum\limits_{i=1}^N w_iw_i'\eta^2 \frac{\norm{\Sigma^{1/2}v}^2}{\norm{\Sigma^{1/2}}_{op}^2} (1 - I_{\eta}(\langle X_i,v\rangle))
\end{align*}

\begin{align*}
    = \frac{1}{M} \sum\limits_{i=1}^N w_iw_i'\eta^2 \frac{\norm{\Sigma^{1/2}v}^2}{\norm{\Sigma^{1/2}}_{op}^2} - \frac{1}{M} \sum\limits_{i=1}^N w_iw_i'\eta^2\frac{\norm{\Sigma^{1/2}v}^2}{\norm{\Sigma^{1/2}}_{op}^2}I_{\eta}(\langle X_i,v\rangle)
\end{align*}

Then using the fact that $I_\eta(\langle X_i,v\rangle)$ is SoS and $\{w_i^2 = w_i\} \SoSp{} (1 - w_i) \SoSge_2 0$, we subtract $\eta^2\frac{\norm{\Sigma^{1/2}v}^2}{\norm{\Sigma^{1/2}}_{op}^2}\sumn w_i'(1-w_i)I_\eta(\langle X_i,v\rangle)$ to obtain 

\begin{align*}
    \geq \frac{1}{M} \sum\limits_{i=1}^N w_iw_i'\eta^2\frac{\norm{\Sigma^{1/2}v}^2}{\norm{\Sigma^{1/2}}_{op}^2} - \eta^2 \frac{\norm{\Sigma^{1/2}v}^2}{\norm{\Sigma^{1/2}}_{op}^2}\frac{1}{M}\sum\limits_{i=1}^N w_i'I_{\eta}(\langle X_i,v\rangle)
\end{align*}

We know from the moment certificates, \ref{eq:anti1}, that  
$$
    \frac{1}{M}\sum\limits_{i=1}^N w_i'I_{\eta}(\langle X_i,v\rangle) 
 = \E_{X \sim N(0,\Sigma)}I_\eta(\iprod{X_i,v}) + O(\epsilon) = \E_{X \sim N(0,I)}I_\eta(\iprod{X,\Sigma^{1/2}v}) + O(\epsilon)
$$

so thus far we have shown, 

\begin{align*}
    \cP \SoSp{D(\eta)}\sum\limits_{i=1}^N w_iw_i'\langle X,v\rangle^2 \geq \frac{1}{M} \sum\limits_{i=1}^N w_iw_i'\eta^2\frac{\norm{\Sigma^{1/2}v}^2}{\norm{\Sigma^{1/2}}_{op}^2} - \eta^2 \frac{\norm{\Sigma^{1/2}v}^2}{\norm{\Sigma^{1/2}}_{op}^2}\E_{x \sim N(0,I)} I_{\eta}(\iprod{X,\Sigma^{1/2}v}) + O(\epsilon)
\end{align*}

For the first term on the right hand side, lower bound $\norm{\Sigma^{1/2}v} \geq \lambda_{d}^2\norm{v}^2$.  This follows by the PSD'ness of $\Sigma^{1/2}$ via degree 2 SoS.  Then change the variable $\omega = \Sigma^{1/2}v$ to obtain 

$$
\geq \frac{1}{M} \sum\limits_{i=1}^N w_iw_i'\eta^2\frac{\lambda_d\norm{v}^2}{\lambda_1} - \eta^2 \frac{\norm{w}^2}{\lambda_1}\E_{x \sim N(0,I)} I_{\eta}(\iprod{X,w}) + O(\epsilon)
$$

Consider the second term.  Observing that $0 \leq \norm{w}^2 \leq \lambda_1$ and scaling \pref{thm:anti-sufficient} condition 2 by $\lambda_1$ we obtain 

$$
\geq \frac{1}{M} \sum\limits_{i=1}^N w_iw_i'\eta^2\frac{\lambda_d\norm{v}^2}{\lambda_1} - \eta^2 (c\eta) + O(\epsilon)
$$

Let $\eta' = \eta\sqrt{\frac{\lambda_d}{\lambda_1}}$, then we conclude 

$$
\geq \frac{1}{M} \sum\limits_{i=1}^N w_iw_i'\eta'^2\norm{v}^2 - \eta'^2 (c\frac{\lambda_1^{3/2}}{\lambda_d^{3/2}}\eta') + O(\epsilon)
$$
as desired. 
\end{proof}
\begin{corollary} (Anticoncentration of Spherically Symmetric Strongly Log Concave Distributions)
Let $p(x_1,...,x_d)$ be a distribution of the form 
$$p(x_1,...,x_d) \propto \exp(-h(\norm{x})) $$
For $h(x)$  $m$-strongly convex.  Then $p(x)$ is $(\sqrt{2em},O(\frac{1}{\eta^4}))$-certifiably anticoncentrated.    
\end{corollary}

\begin{proof}
The proof follows exactly as that of the Gaussian.  
We begin with 
$$H(\norm{v})  := \norm{v}^2\E_{x \sim p(x)} I_\eta(\iprod{X,v}) = \norm{v}^2\E_{x \sim p(\frac{x}{\norm{v}})\frac{1}{\norm{v}}} I_\eta(x)$$

Applying $m$-strong concavity we obtain 

\begin{align*}
    \norm{v}^2\E_{x \sim p(\frac{x}{\norm{v}})\frac{1}{\norm{v}}} f(x) \leq \norm{v}^2\int \frac{\sqrt{em/2}}{\sqrt{2\pi}\norm{v}} \exp(-\frac{x^2}{2\eta^2} - \frac{mx^2}{\norm{v}^2}) dx
    \leq \norm{v}^2 \frac{\eta\rho \sqrt{em/2}}{\norm{v}} \leq \eta\sqrt{em/2}
\end{align*}
With the polynomial approximation calculations following the exact same template. 
\end{proof}

\bibliographystyle{amsalpha}
\bibliography{bib/mr,bib/dblp,bib/scholar,bib/bibliography,bib/listDecoding}

\appendix

\section{Regression Missing Proofs}
\restatelemma{lem:roundtree}
\begin{proof}
There are a variety of techniques for boosting the success probability to $1 - \frac{1}{\poly(d)}$.  One such technique is to make the rounding algorithm deterministic.  Instead of using selection strategy $j \in \cS_v$, simply condition on a variable $j \in [N]$ satisfying
 \begin{align*}
        \norm{\cQ}_\text{nuc}  - \E_{b_j}\norm{\cQ\big|_{w_j = b_j}}_\text{nuc} \geq
        \Omega(\beta\eta^2\rho^2)
\end{align*}
Such a variable necessarily exists, because we found a distribution over $j$ where the above inequality holds in expectation.     
Furthermore, enumerate every $\{0,1\}$ conditioning up to a depth of $R = O(\frac{1}{\beta^4})$. This implicitly defines a tree of pseudoexpectations.  We can compute the probability of reaching each leaf via its $\{0,1\}$ conditioning sequence.  In effect, we can compute a probability distribution over a list $L$ of $2^R$ estimates to $\ell'$.  Then applying the same analysis in \pref{thm:main-regression-estimation} the probability over this distribution that $\ell_i \in L$ is close to $\ell'$ is greater than $\frac{\beta}{4}$.  A simple clustering algorithm which groups vectors endowed with high probability mass generates a list of length $O(\frac{1}{\beta})$.  This procedure is deterministic, and can be run changing neither the runtime nor error gaurantees.          
\end{proof}

\restatelemma{lem:reg-approx-round}
The proofs of \ref{rr5} and \ref{rr6} are nearly identical.  We include both below.
\begin{proof}(Proof of \ref{rr5})

Let $u$ be a unit direction $u \in \cS^{d-1}$. We know

\begin{align*}
    \pVar\left[\langle \ell,u\rangle\right] = \pVar\left[\ell^T\left(I - \sumn w_i
    X_iX_i^T\right) + \sumn w_i\langle \ell,X_i\rangle X_i^T u\right]
\end{align*}
\Pnote{we should fix slash left bracket, slash right bracket in all of calculations below}
Using pseudovariance triangle inequality

\begin{equation}\label{rr1}
    \leq  (1 + \psi)\pVar\left[\sumn w_i\langle \ell,X_i \rangle\langle X_i, u \rangle\right] + \Big(\frac{1 + \psi}{\psi}\Big)\pVar\left[\langle \ell u^T,I - \sumn w_iX_iX_i ^T)\rangle\right ]
\end{equation}

Applying $\ell_2$ minimization constraint 6 to the first term we obtain \Pnote{ numbering SDP constraints, and also refering to them by number makes it more readable here}
\begin{equation}\label{rr1}
    = (1 + \psi)\pVar\left[\sumn w_i y_i\langle X_i, u \rangle\right] + \Big(\frac{1 + \psi}{\psi}\Big)\pVar\left[\langle \ell u^T,I - \sumn w_iX_iX_i ^T)\rangle\right]
\end{equation}

Consider the second term on the right hand side, which we upper bound as follows. 

\begin{align*}
    \pVar\left[\langle \ell u^T,I - \sumn w_iX_iX_i ^T)\rangle\right] \leq \pE\left[\langle \ell u^T,I - \sumn w_iX_i X_i^T)\rangle^2\right]
\end{align*}

First using deg $2$ SOS Cauchy-Schwarz, and then both the scaling constraint  and the moment bound constraint we obtain, 

\begin{equation}\label{rr7}
    \leq 2\pE\left[\norm{\ell}^2\norm{I - \sumn w_iX_iX_i^T}^2\right] \leq \rho^2\epsilon
\end{equation}

Plugging back into equation \ref{rr1} we obtain  
\begin{equation}\label{rr2}
    \pVar\left[\langle \ell,u\rangle\right] \leq  (1+\psi)\pVar\left[\sumn w_iy_i\langle X_i, u \rangle\right] + \Big(\frac{1+\psi}{\psi}\Big)\rho^2\epsilon
\end{equation}

\begin{align*}
    =  (1 + \psi)\pVar\left[\sumn w_i (y_i - \langle \pE[\ell],X_i \rangle)\langle X_i, u \rangle + \sumn w_i \langle \pE[\ell],X_i\rangle\langle X_i, u \rangle \right] + \frac{1 + \psi}{\psi}\rho^2\epsilon
\end{align*}

Using the pseudovariance triangle inequality we obtain

\begin{equation}\label{rr3}
    \begin{aligned}
    \leq  (1 + \psi)^2\pVar\left[\sumn w_i(y_i - \langle \pE[\ell],X_i\rangle)\langle X_i,u\rangle\right]+ \frac{(1+\psi)^2}{\psi}\pVar\left[ \sumn w_i\langle \pE[\ell],X_i\rangle \langle X_i,u\rangle\right]   \\ + \frac{1 + \psi}{\psi}\rho^2\epsilon
\end{aligned}
\end{equation}
Consider the second term.  Subtracting the identity and adding it back we obtain

\begin{align*}
    \pVar\left[ \sumn w_i\langle \pE[\ell],X_i\rangle \langle X_i ,u\rangle\right] \leq \pVar\left[\langle \pE[\ell] u^T, \sumn w_iX_i X_i^T-I\rangle + \langle \pE[\ell],u\rangle \right]
\end{align*}

Noting that pseudovariance is invariant under constant shifts we obtain.
\begin{align*}
    \leq \pVar\left[\langle \pE[\ell] u^T, \sumn w_iX_iX_i ^T-I\rangle\right]
    \leq \pE\left[\langle \pE[\ell]u^T, \sumn w_iX_iX_i^T - I\rangle^2\right]
\end{align*}

Applying Cauchy-Schwarz, then the moment constraints, then pseudoexpectation Cauchy-Schwarz, then the scaling constraints we obtain  

\begin{equation}\label{rr8}
    \leq \pE\left[\norm{\pE[\ell]}^2 \norm{\sumn w_iX_iX_i^T - I}_F^2\right] \leq \norm{\pE[\ell]}^2 \epsilon \leq  \pE[\norm{\ell}^2]\epsilon \leq \rho^2\epsilon
\end{equation}

Plugging this bound back into \ref{rr3} we obtain 

\begin{align*}
     \pVar\left[\iprod{\ell,u}\right]
     \leq   (1+\psi)^2\pVar\left[\sumn w_i(y_i - \langle \pE[\ell],X_i\rangle)\langle X_i, u \rangle\right] + \frac{(1+\psi)^2}{\psi}\rho^2\epsilon + \frac{1+\psi}{\psi}\rho^2\epsilon
\end{align*}
For $\psi = \sqrt{\epsilon}$ we obtain 
\begin{equation}\label{rr5}
    \pVar[\langle \ell,u\rangle] \leq (1 + o_d(1))\pVar[\sumn w_i(y_i - \langle \pE[\ell],X_i\rangle)\langle X_i, u \rangle] +  o_d(1) 
\end{equation}
\end{proof}

\begin{proof}(Proof of \ref{rr6})
\begin{align*}
     \pVar\left[\sumn w_i(y_i - \langle \pE[\ell],X_i\rangle)\langle X_i, u \rangle\right]
     =   \pVar\left[\sumn w_iy_i\langle X_i, u \rangle  - \sumn w_i\iprod{\pE[\ell],X_i}\langle X_i, u \rangle\right]
\end{align*}

by pseudovariance triangle inequality

\begin{align*}
     \leq (1+\psi)\pVar\left[\sumn w_iy_i\langle X_i, u \rangle\right] + \frac{1+\psi}{\psi}\pVar\left[ \sumn w_i\iprod{\pE[\ell],X_i}\langle X_i, u \rangle\right]
\end{align*}

applying the bound in \ref{rr8}
\begin{align*}
     \leq (1+\psi)\pVar\left[\sumn w_iy_i\langle X_i, u \rangle\right] + \frac{1+\psi}{\psi}\rho^2\epsilon
\end{align*}

Applying the $\ell_2$ minimization constraint we obtain 
\begin{align*}
     = (1+\psi)\pVar\left[\sumn w_i
     \iprod{\ell,X_i}\langle X_i, u \rangle\right] + \frac{1+\psi}{\psi}\rho^2\epsilon
\end{align*}

\begin{align*}
     = (1+\psi)\pVar\left[\iprod{\ell u^T,\sumn w_iX_iX_i^T - I} + \iprod{\ell,u}\right] + \frac{1+\psi}{\psi}\rho^2\epsilon
\end{align*}
 then by pseudovariance triangle inequality 
 
 \begin{align*}
     \leq  (1+\psi)^2\pVar\left[\iprod{\ell,u}\right] + \frac{(1+\psi)^2}{\psi}\pVar\left[\iprod{\ell u^T,\sumn w_iX_iX_i^T - I}\right] + \frac{1+\psi}{\psi}\rho^2\epsilon
\end{align*}

Applying the bound in \ref{rr7} to the second term we conclude 

\begin{align*}
     \pVar\left[\sumn w_i(y_i - \langle \pE[\ell],X_i\rangle)\langle X_i, u \rangle\right]
     \leq   (1+\psi)^2\pVar[\iprod{\ell,u}] + \frac{(1+\psi)^2}{\psi}\rho^2\epsilon + \frac{1+\psi}{\psi}\rho^2\epsilon
\end{align*}
For $\psi = \sqrt{\epsilon}$ we conclude 
\begin{equation}\label{rr6}
 \pVar[\sumn w_i(y_i - \langle \pE[\ell],X_i\rangle)\langle X_i, u \rangle]\leq (1 + o_d(1))\pVar[\langle \ell,u\rangle] + o_d(1)
\end{equation}
as desired.

\end{proof}

\section{Mean Algorithms}

\subsection*{SDP for Robust Mean Estimation}

Here we write down the list decoding algorithm for mean estimation.  Let RobustMeanSDP($\cD,\rho$) take as input the dataset $\cD$, and the parameter $\rho$. 

\begin{algorithm}[H] \label{algo:robustregressionSDP}
\SetAlgoLined
\KwResult{A degree $D$ pseudoexpectation functional $\pE_\zeta$}
 \textbf{Inputs}: $(\cD,\rho)$ Dataset, and upper bound on $\norm{\mu}$\\
 \begin{equation*}
\begin{aligned}
& \underset{\text{degree D pseudoexpectations} \pE}{\text{minimize}}
& & \sum_{i=1}^N \pE[w_i]^2 \\
& \text{such that for all }
& & \pE[(w_i^2 - w_i)] = 0, \; i \in [N], \\
& & & \pE[(\sum_{i=1}^N w_i - M)]  = 0, \\
& & & \pE[||(\frac{1}{M}\sum_{i=1}^N w_i (X_i - \hat{\mu})^{\otimes \frac{k}{2}}((X_i - \hat{\mu})^{\otimes \frac{k}{2}})^T - M_k||_F^2 - \epsilon_k)] \leq 0  ,\\
& & & \pE[(\norm{\hu}^2 - \rho^2)] \leq 0\\
\end{aligned}
\end{equation*}\\
 \textbf{return}: $\pE_\zeta$
 \caption{RobustMeanSDP}
\end{algorithm}

\subsection*{Algorithms for Robust Mean Estimation}
The algorithms are identical to those of robust regression up to parameter choices and the choice of rounding strategy $\cS$.

\begin{algorithm}[H] \label{algo:meanroundtree}
\SetAlgoLined
\KwResult{A list $L$ of means of length $O(\frac{1}{\beta})$}

 \textbf{inputs}: $\pE_\zeta$\\
 $L = \emptyset$ \\
  \For{$j = 1$ to $N$}{
    $L = L \cup (\pE[\hu|w_j=1],\pE[w_j])$\\
  }
 \textbf{return:} ExtractList(L)
 \caption{Mean Preprocessing Algorithm}
\end{algorithm}

 \begin{algorithm}[H] \label{algo:extract}
\SetAlgoLined
\KwResult{A list of vectors $L = \{\mu_1,...,\mu_A\}$}
 \textbf{Inputs}: A list $L' := \{(\mu_i,p_i)\}$\\
 $L = \emptyset$\\
 \While{$L' \neq \emptyset$ }{
    Let $\mu_0$ be any leaf in $L'$\\
    \For{$(\mu,p) \in L'$}{
        $M = 0$\\
        \If{$\norm{\mu - \mu_0 }\leq \frac{\rho}{4}$}{
            $M = M + p$\\
        }
    }
    \eIf{$M \geq \frac{\beta}{4}$}{
        $L = L\cup \mu_0$\\    
    }{
        $L = L'\backslash \mu_0$
    }
 }
 \textbf{return:} L
 \caption{Extract List}
 \end{algorithm}

\newpage
\begin{algorithm}[H] \label{algo:meanroundtree}
\SetAlgoLined
\KwResult{A set of leaves $\cL = \{T_h\}$ indexed by their position $h$.  Each leaf is a tuple $T_h = (\Pr[h],\pE_{\zeta,h})$ consisting of the probability assigned to $T_h$, and the corresponding pseudoexpectation at $T_h$}
 \textbf{inputs}: $(\pE_{\zeta,h}, \Pr[h],\rho)$\\
 Let $\cQ = \pE_{\zeta,h}[(\hu - \pE_{\zeta,h}[\hu])(\hu - \pE_{\zeta,h}[\hu])^T]$ be the pseudocovariance matrix of the estimator $\hu$\\
 \text{Let }$(\lambda,v)$\text{ be top eigenvalue/vector of $\cQ$} \\
 \eIf{ $\lambda > \frac{\beta\rho^2}{16}$}{
  Let $\cS_v$ be a probability distribution over $[N]$\;
  Where for any $j \in N$ we have $\cS_v(j) \defeq \frac{\pVar_{\zeta,h}[w_j\langle X_j-\pE_{\zeta,h}[\hu],v\rangle]}{\sum_{i=1}^N \pVar_{\zeta,h}[w_j\langle X_j-\pE_{\zeta,h}[\hu],v\rangle]}$ \\
  Sample $j \sim \mu$ \\
  Let $\Pr[h \circ 1] = \Pr[h]\pE_{\zeta,h}[w_j]$ \\
  Let $\Pr[h \circ 0] = \Pr[h]\pE_{\zeta,h}[1 - w_j]$\\
  Let $\pE_{\zeta,h \circ 0} = \pE_{\zeta,h}\big|_{w_j = 0}$\\
  Let $\pE_{\zeta,h \circ 1} = \pE_{\zeta,h}\big|_{w_j = 1}$\\
  \textbf{return:} $\{\text{RoundTree}(\pE_{\zeta,h \circ 0}, \Pr[h \circ 0],\rho)\} \cup \{\text{RoundTree}(\pE_{\zeta,h \circ 1}, \Pr[h \circ 1],\rho)\}$ \;
 }{
 \textbf{return:} $\{(\pE_{\zeta,h}[\hu],\Pr[h],\rho)\}$
 }
 \caption{Mean RoundTree Algorithm}
\end{algorithm}

\begin{algorithm}[H]
\SetAlgoLined
\KwResult{A list of mean $L = \{\mu_1,...,\mu_A\}$}
 \textbf{inputs}: $(\pE_\zeta,\rho)$\\
 $L = \text{MeanPreprocessing}(\pE_\zeta)$\\
 \For{ $t \in \log(\frac{m_2}{\beta^{1 - 1/k}})$}{
    $\rho_t = \frac{\rho}{2^t}$\\
    \% Let $Y$ to be a list of pseudoexpectations\\
    $Y = \emptyset$ \\
    \For{$\mu_i \in L$}{
        \For{$X_j \in \cD$}{
            $X_j = X_j - \mu_i$\\
        }
        Let $Y = Y \cup \text{RobustMeanSDP}(\cD,\mu_i)$\\
    }
    $L = \emptyset$\\
    \For{$\pE_\zeta \in Y$}{
        $L' = \text{RoundTree}(\pE_\zeta,1,\frac{\rho}{2})$\\
        $L = L \cup \text{ExtractList}(L')$
    }
 }
 \textbf{return:} $L$
 \caption{ListDecodeMean Algorithm} \label{algo:meaniterative}
\end{algorithm}
\section{Mean Estimation Overview}\label{sec:mean-estimation}

A convenient feature of our list decoding framework is that any setting for which we can prove "variance reduction" and "snapping" gives us a list decoding algorithm.  We prove the analogues \pref{lem:regstrategy} and \pref{lem:regressionSnapping} for the setting of mean estimation.    
\subsection{Mean Estimation}

In this section, we will lay out the broad overview of the proof of our algorithm for mean estimation.  Specifically, we will show the following.

\begin{theorem} \label{thm:main-mean-estimation}
    ListDecodeMean($\cD,\rho$) outputs a list $L = \{\mu_1,...,\mu_A\}$ of length $|L| = O\big(\frac{1}{\beta}^{\log(\frac{m_2}{\beta^{1 - 1/k}})}\big)$ such that for some   $\hu \in L$ we have 
    
    $$ \norm{\mu' - \hu} \leq O\Big( \frac{1}{\beta^{1/k}}\Big)$$
    with high probability.  The algorithm runs in time $ (\frac{1}{\beta})^{\log(\frac{m_2}{\beta^{1 - 1/k}})}d^{O(\max(k,\frac{1}{\beta^4}))}$
\end{theorem}
\pref{thm:main-mean-estimation} gives us recovery guarantees that are information theoretically optimal.  A straightforward post-processing step takes the list length down to $O(\frac{1}{\beta})$.  See $\cite{diakonikolas2018list}$ Appendix B proposition B.1.  

To prove \pref{thm:main-mean-estimation} we will need the following lemmas. 

\begin{lemma}\torestate{\label{lem:meansnapping} (Snapping)
Suppose $\pE_\zeta$ is a degree $k$-pseudoexpectation operator that satisfies the constraints of RobustMean SOS SDP.  If the variance is small 
$$ \max_u \pVar[\iprod{\hat{\mu},u}] \leq \frac{\beta\rho^2}{16} $$ 
and there is correlation with the plant 
$$ \pE[ \sumn w_i w'_i] \geq \beta $$
then,
$$ \norm{\pE_\zeta[\hu] - \mu} \leq O(\frac{1}{\beta^{1/k}})$$}
\end{lemma}
\begin{lemma} \torestate{
\label{lem:cr2} 
Let $\pE_\zeta$ satisfy the pseudoexpectation constraints of RobustMeanSDP($\cD,\rho$). $\cQ$ be the pseudocovariance matrix. Let $(\lambda,v)$ be the largest eigenpair of $\cQ$.  
    Let $\cS_v$ be a probability distribution over $[N]$\;
  Where for any $j \in N$ we have 
$$\cS_v(j) \defeq \frac{\pVar[w_j\iprod{x_j-\pE[\hu],v}]}{\sum_{i=1}^N \pVar[w_j\iprod{x_j-\pE[\hu],v}]} $$

Then
\begin{align*}
\pVar[\iprod{\hu,v}] - \E_{j \sim S} \E_{b_j} \pVar[\iprod{\hu,v}|w_j = b_j] \geq \frac{\beta}{2}\frac{\pVar[\iprod{\hu,v}]^2}{2(\pE[\frac{1}{ N} \sum\limits_{i=1}^N w_i\iprod{X_i - \hu,v}^2] + \pVar[\iprod{\hu,v}])}  
\end{align*}}
\end{lemma}

Now we can prove \pref{thm:main-mean-estimation}

\begin{proof}
For $\norm{\cQ}_\text{op} \geq \Omega(\beta\rho^2)$,  we have by \pref{cor:nnr} a nuclear norm decrease $\phi = \Omega(\beta^3\rho^2)$.  Then by \pref{lem:highdimensionalvariance} we have an algorithm that runs in time $d^{O(\max(k,\frac{1}{\beta^4}))}$.  Iterating the RoundTree algorithm as we do for regression gives us the final error gaurantee and list length.  
\end{proof}
\section{Mean Lemmas} \label{sec:meansnapping}


\begin{proof}
\prettyref{lem:cr2}

We invoke \pref{thm:onedimrounding} for $z_i \seteq w_i\iprod{X_i - \pE[\hu],v}$ for $i \in [N]$. 
Now it suffices to upper bound $\frac{1}{\beta N}\sum\limits_{i=1}^N \pVar(z_i)$.
\begin{align*}
    \sumn \pVar(z_i) \leq \sumn \pE[z_i^2] =   \pE\left[\sumn w_i\iprod{X_i - \E[\hu],v}^2\right]
\end{align*}
We add and subtract $\hu$ and expand the expression with triangle inequality to obtain
\begin{align*}
     =   \pE\left[\sumn w_i\iprod{X_i - \hu + \hu -  \E\hu,v}^2\right] = \pE\left[\sumn w_i\iprod{X_i - \hu + \hu -  \E\hu,v}^2\right]
\end{align*}

\begin{align*}
      \leq 2\Big(\pE\left[\frac{1}{N}\sumn w_i\iprod{X_i - \hu,v}^2\right] + \pE\left[\iprod{\hu -  \pE\hu,v}^2\right]\Big) = 2\Big(\pE\left[\sumn w_i\iprod{X_i - \hu,v}^2\right] + \pVar[\iprod{\hu,v}]\Big)
\end{align*}
as desired.
\end{proof}


\begin{proof} \pref{lem:meansnapping}

First, we will make the following claim which we will prove later
\begin{claim}\label{mc}
\begin{align*}
    \pE\left[\sumn w_iw_i'\right](\langle\pE[\hu] - \mu,u\rangle)^2 \leq \pE\left[\langle\hu - \pE\hu,u\rangle^2\right] + \pE\left[\sumn w_iw_i' \norm{\hu - \mu}^2\right]
\end{align*}
\end{claim}
Which implies
\begin{align*}
    \langle\pE[\hu] - \mu,u\rangle^2\leq \frac{\pE[\langle\hu - \E[\hu],u\rangle^2]}{\pE\left[\sumn w_iw_i'\right]} + \pE\left[\frac{\sumn w_iw_i'}{\pE\left[\sumn w_iw_i'\right]} \norm{\hu - \mu}^2\right]
\end{align*}

Now we bound the second term on the right hand side.
Let $\bcalG(w)$ be the space of polynomials over $w_1,...,w_N$.  Consider the functional $\pE': \bcalG(w) \rightarrow \R$ which takes a polynomial $f(w)$ and maps it to $\pE' f(w) = \pE\left[\frac{\sumn w_iw_i'}{\pE\left[\sumn w_iw_i'\right]}f(w)\right]$.  We observe that $\pE'$ is a valid pseudoexpectation.  Therefore, applying cauchy-schwarz we obtain

\begin{align*}
\pE\left[\frac{\sumn w_iw_i'}{\pE\left[\sumn w_iw_i'\right]}\norm{\hu - \mu}^2\right] = \pE'\left[ \norm{\hu - \mu}^2\right] \leq \pE'\left[\norm{\hu - \mu}^k\right]^{\frac{2}{k}} = \pE[\frac{\sumn w_iw_i'}{\pE[\sumn w_iw_i']}||\hu - \mu||^k]^{2/k}
\end{align*}
Rearranging we have
\begin{align*}
    \pE\left[\sumn w_iw_i'\norm{\hu - \mu}^2\right] \leq \pE\left[\sumn w_iw_i'\norm{\hu - \mu}^k\right]^\frac{2}{k}\pE\left[\sumn w_iw_i'\right]^{1 - \frac{2}{k}}
\end{align*}
Then by SOS Cauchy-Schwarz and then the SOS inequality \cite{hopkins2018mixture} lemma 5.5 we have 
\begin{align*}
    \pE\left[\sumn w_iw_i'\norm{\hu - \mu}^k\right] \leq \pE\left[(\sumn w_iw_i')^2\norm{\hu - \mu}^{2k}\right]^\frac{1}{2} \leq \pE\left[\sumn w,w_i\norm{\hu - \mu}^k\right]^{\frac{1}{2}}
\end{align*}
Any number that is smaller than its square root is less than one i.e
\begin{align*}
    \pE\left[\sumn w_iw_i'\norm{\hu - \mu}^k\right] \leq 1
\end{align*}
So we conclude that
\begin{align*}
    \pE\left[\sumn w_iw_i'\norm{\hu - \mu}^2\right] \leq \pE\left[\sumn w_iw_i'\right]^{1 - \frac{2}{k}}
\end{align*}

Plugging back into the second term above we obtain

\begin{align*}
    \leq \frac{\pE\left[\sumn w_iw_i'\right]^{1 - \frac{2}{k} }}{ \pE\left[\sumn w_iw_i'\right]} = \frac{1}{\pE\left[\sumn w_iw_i'\right]^\frac{2}{k}} \leq \frac{1}{\beta^{\frac{2}{k}}}
\end{align*}
Where in the last line we used the result of frobenius norm minimization $\pE\left[\sumn w_iw_i'\right] \geq \beta$. Then by taking square root on both sides we obtain
\begin{align*}
    \norm{\E[\hu] - \mu} \leq 1 + \frac{1}{\beta^{\frac{1}{k}}} = O(\frac{1}{\beta^{\frac{1}{k}}})
\end{align*}

\end{proof}

Now we prove the claim.
\begin{proof}\ref{mc}
\begin{align*}
    \pE\left[\sumn w_iw_i'\right]\langle \pE[\hu] - \mu,u\rangle^2 = \pE\left[\sumn w_iw_i'\langle \pE[\hu] - \mu,u\rangle^2\right]
\end{align*}
Applying triangle inequality 
\begin{align*}
    =\pE\left[\sumn w_iw_i'\langle \pE[\hu] - \hu + \hu - \mu,u\rangle^2\right] \leq 2
    \pE\left[\sumn w_iw_i'\langle \pE[\hu] - \hu,u\rangle^2\right]  + 2\pE\left[\sumn w_iw_i'\langle \hu - \mu,u\rangle^2\right]
\end{align*}

Using  $w_iw_i' \leq 1$ in the first term, we have

\begin{align*}
    \leq 2
    \pE\left[\langle \pE[\hu] - \hu,u\rangle^2\right]  + 2\pE\left[\sumn w_iw_i'\langle \hu - \mu,u\rangle^2\right]
\end{align*}

by SOS Cauchy-Schwarz on the second term we obtain 

\begin{align*}
    \pE\left[\sumn w_iw_i'\right]\langle \pE[\hu] - \mu,u\rangle^2 \leq 2
    \pE\left[\langle \pE[\hu] - \hu,u\rangle^2\right]  + 2\pE\left[\sumn w_iw_i'\norm{\hu - \mu}^2\right]
\end{align*}

\end{proof}

\subsection{Preprocessing via Conditioning}
\begin{theorem}
Consider the strategy $\cS$ of conditioning on $w_j = 1$ where $j$ is selected with probability 
$\cS(j) = \frac{\pE[w_j]}{M}$.  Then in expectation over the selection of $j$, we have that $\norm{\cQ}_{op}$ is small in expectation.
\begin{align*}
    \E_{j \sim \cS}\pVar(\iprod{\hu,u}|w_j = 1) < 4m_2
\end{align*}

and there is correlation with the plant.
\begin{align*}
    \E_{j \sim \cS}\pE[\sumn w_iw_i'|w_j = 1] \geq \beta
\end{align*}
\end{theorem}

\begin{proof}
The correlation with the plant follows by the definition of $\cS$. We now upper bound $\norm{\cQ}_{op}$ explicitly.  
\begin{align*}
\E_{j \sim \cS} \pVar(\iprod{\hu,u}|w_j = 1) = \E_{j \sim \cS}\pE[\iprod{\hu - \pE[\hu|w_j=1],u}^2|w_j=1]
\end{align*}
by triangle inequality
\begin{align*}
    \leq 2\E_{j \sim \cS} \pE[\iprod{\hu - X_j,u}^2|w_j=1] + 2\pE_{j \sim \cS}\iprod{X_j - \pE[\hu|w_j=1],u}^2
\end{align*}

Applying the definition of $\cS$ to the first term we obtain 

\begin{align*}
    = 2\pE[\frac{1}{M}\sum\limits_{j=1}^N w_j\iprod{X_j - \hu,u}^2] + 2\E_{j \sim \cS}\iprod{X_j - \pE[\hu|w_j=1],u}^2
\end{align*}

Then using pseudoexpectation cauchy-schwarz on the second term we obtain 
\begin{align*}
    \leq 2\pE[\frac{1}{M}\sum\limits_{j=1}^N w_j\iprod{X_j - \hu,u}^2] + 2\E_{j \sim \cS}\pE[\iprod{X_j - \hu,u}^2|w_j=1]
\end{align*}

Applying the definition of $\cS$ to the second term we obtain  

\begin{align*}
    \leq 2\pE[\frac{1}{M}\sum\limits_{j=1}^N w_j\iprod{X_j - \hu,u}^2] + 2\pE[\frac{1}{M}\sum\limits_{j=1}^Nw_j\iprod{X_j - \hu,u}^2] \leq 4m_2
\end{align*}

as desired.  
\end{proof}

\end{document}